\documentclass[runningheads,orivec]{llncs}

\usepackage{amsmath,amssymb}
\usepackage{stmaryrd}
\usepackage{ifthen}
\usepackage{pgfplots}
\pgfplotsset{compat=1.5, width=10cm}
\usepackage{subfig}

\allowdisplaybreaks

\newtheorem{construction}{Construction}

\newcommand{\dshift}{\!\upharpoonright\!}
\newcommand{\shift}{\!\mapsfrom\!}

\newcommand{\shadecolor}{black!20}

\newcommand{\B}{\ensuremath{\mathcal{B}}}
\newcommand{\C}{\ensuremath{\mathcal{C}}}

\renewcommand{\L}{\ensuremath{\mathcal{L}}}

\newcommand{\R}{\ensuremath{\mathcal{R}}}
\renewcommand{\S}{\ensuremath{\mathcal{S}}}
\newcommand{\T}{\ensuremath{\mathcal{T}}}

\newcommand{\W}{\ensuremath{\mathcal{W}}}
\newcommand{\Z}{\ensuremath{\mathcal{Z}}}

\newcommand{\North}{\ensuremath{\vec{\text{\sc n}}}}
\newcommand{\East}{\ensuremath{\vec{\text{\sc e}}}}
\newcommand{\South}{\ensuremath{\vec{\text{\sc s}}}}
\newcommand{\West}{\ensuremath{\vec{\text{\sc w}}}}
\newcommand{\north}{\ensuremath{\vec{\text{\sc n}}}}
\newcommand{\east}{\ensuremath{\vec{\text{\sc e}}}}
\newcommand{\south}{\ensuremath{\vec{\text{\sc s}}}}
\newcommand{\west}{\ensuremath{\vec{\text{\sc w}}}}

\newcommand{\gcolor}{\ensuremath{{color}}}
\newcommand{\tcolor}{\gcolor} 
\newcommand{\gstrength}{\ensuremath{{strength}}}
\newcommand{\tstrength}{\gstrength} 
\newcommand{\tlabel}{\ensuremath{{\ell}}}

\newcommand{\row}{\ensuremath{{row}}}
\newcommand{\col}{\ensuremath{{col}}}
\newcommand{\block}{\ensuremath{{block}}}
\newcommand{\first}{\ensuremath{{first}}}
\newcommand{\second}{\ensuremath{{second}}}

\usepackage{tikz}
\usetikzlibrary{positioning} 
\usetikzlibrary{arrows}
\usetikzlibrary{shapes.misc} 

\newcommand{\tile}[8]{
	\fill[#3] (#1+.05,#2+.05) rectangle (#1+1.95,#2+1.95);
	\draw[black]  (#1+.057,#2+.057) -- (#1+1.945,#2+.057);
	\draw[black]  (#1+1.9422,#2+.05) -- (#1+1.9422,#2+1.9493);
	\draw[black]  (#1+.055,#2+1.9423) -- (#1+1.945,#2+1.9423);
	\draw[black]  (#1+.0577,#2+.0505) -- (#1+.0577,#2+1.9493);
	
	 \node at (#1+1,#2+1) {\footnotesize #4};
	 
	 \node[label={[label distance=-1.55mm]below:\rotatebox{0}{\scriptsize #5}}] at (#1+1,#2+2) {};
	 \node[label={[label distance=-1.55mm]left:\rotatebox{270}{\scriptsize #6}}] at (#1+2,#2+1) {};
	 \node[label={[label distance=-1.55mm]above:\rotatebox{0}{\scriptsize #7}}] at (#1+1,#2) {};
	 \node[label={[label distance=-1.55mm]right:\rotatebox{270}{\scriptsize #8}}] at (#1,#2+1) {}; 
}

\newcommand{\MEDIUMtile}[8]{
	\fill[#3] (#1+.05,#2+.05) rectangle (#1+1.95,#2+1.95);
	\draw[black]  (#1+.057,#2+.057) -- (#1+1.945,#2+.057);
	\draw[black]  (#1+1.9422,#2+.05) -- (#1+1.9422,#2+1.9493);
	\draw[black]  (#1+.055,#2+1.9423) -- (#1+1.945,#2+1.9423);
	\draw[black]  (#1+.0577,#2+.0505) -- (#1+.0577,#2+1.9493);
	
	 \node at (#1+1,#2+1) {\normalsize #4};
	 
	 \node[label={[label distance=-1mm]below:\rotatebox{0}{\small #5}}] at (#1+1,#2+2) {};
	 \node[label={[label distance=-1mm]left:\rotatebox{270}{\small #6}}] at (#1+2,#2+1) {};
	 \node[label={[label distance=-1mm]above:\rotatebox{0}{\small #7}}] at (#1+1,#2) {};
	 \node[label={[label distance=-1mm]right:\rotatebox{270}{\small #8}}] at (#1,#2+1) {}; 
}

\newcommand{\LARGEtile}[8]{
	\fill[#3] (#1+.05,#2+.05) rectangle (#1+1.95,#2+1.95);
	\draw[black]  (#1+.057,#2+.057) -- (#1+1.945,#2+.057);
	\draw[black]  (#1+1.9422,#2+.05) -- (#1+1.9422,#2+1.9493);
	\draw[black]  (#1+.055,#2+1.9423) -- (#1+1.945,#2+1.9423);
	\draw[black]  (#1+.0577,#2+.0505) -- (#1+.0577,#2+1.9493);
	
	 \node at (#1+1,#2+1) {\normalsize #4};
	 
	 \node[label={[label distance=0mm]below:\rotatebox{0}{\small #5}}] at (#1+1,#2+2) {};
	 \node[label={[label distance=0mm]left:\rotatebox{270}{\small #6}}] at (#1+2,#2+1) {};
	 \node[label={[label distance=0mm]above:\rotatebox{0}{\small #7}}] at (#1+1,#2) {};
	 \node[label={[label distance=0mm]right:\rotatebox{270}{\small #8}}] at (#1,#2+1) {}; 
}

\newcommand{\seed}[8]{
	\fill[#3] (#1+.05,#2+.05) rectangle (#1+1.95,#2+1.95);
	\draw[black,dotted]  (#1+.072,#2+.057) -- (#1+1.922,#2+.057);
	\draw[black,ultra thick]  (#1+1.922,#2+.05) -- (#1+1.922,#2+1.95);
	\draw[black,ultra thick]  (#1+.05,#2+1.922) -- (#1+1.95,#2+1.922);
	\draw[black,dotted]  (#1+.057,#2+.072) -- (#1+.057,#2+1.922);
	
	 \node at (#1+1,#2+1) {\footnotesize #4};
	 
	 \node[label={[label distance=-1.55mm]below:\rotatebox{0}{\scriptsize #5}}] at (#1+1,#2+2) {};
	 \node[label={[label distance=-1.55mm]left:\rotatebox{270}{\scriptsize #6}}] at (#1+2,#2+1) {};
	 \node[label={[label distance=-1.55mm]above:\rotatebox{0}{\scriptsize #7}}] at (#1+1,#2) {};
	 \node[label={[label distance=-1.55mm]right:\rotatebox{270}{\scriptsize #8}}] at (#1,#2+1) {}; 
}

\newcommand{\vtile}[8]{
	\fill[#3] (#1+.05,#2+.05) rectangle (#1+1.95,#2+1.95);
	\draw[black, ultra thick]  (#1+.05,#2+.078) -- (#1+1.945,#2+.078);
	\draw[black]  (#1+1.943,#2+.05) -- (#1+1.943,#2+1.9493);
	\draw[black,ultra thick]  (#1+.05,#2+1.922) -- (#1+1.95,#2+1.922);
	\draw[black,dotted]  (#1+.057,#2+.072) -- (#1+.057,#2+1.922);
	
	 \node at (#1+1,#2+1) {\footnotesize #4};
	 
	 \node[label={[label distance=-1.55mm]below:\rotatebox{0}{\scriptsize #5}}] at (#1+1,#2+2) {};
	 \node[label={[label distance=-1.55mm]left:\rotatebox{270}{\scriptsize #6}}] at (#1+2,#2+1) {};
	 \node[label={[label distance=-1.55mm]above:\rotatebox{0}{\scriptsize #7}}] at (#1+1,#2) {};
	 \node[label={[label distance=-1.55mm]right:\rotatebox{270}{\scriptsize #8}}] at (#1,#2+1) {}; 
}

\newcommand{\htile}[8]{
	\fill[#3] (#1+.05,#2+.05) rectangle (#1+1.95,#2+1.95);
	\draw[black,dotted]  (#1+.072,#2+.057) -- (#1+1.922,#2+.057);
	\draw[black,ultra thick]  (#1+1.922,#2+.05) -- (#1+1.922,#2+1.95);
	\draw[black]  (#1+.055,#2+1.9423) -- (#1+1.945,#2+1.9423);
	\draw[black,ultra thick]  (#1+.077,#2+.0505) -- (#1+.077,#2+1.9493);
	
	 \node at (#1+1,#2+1) {\footnotesize #4};
	 
	 \node[label={[label distance=-1.55mm]below:\rotatebox{0}{\scriptsize #5}}] at (#1+1,#2+2) {};
	 \node[label={[label distance=-1.55mm]left:\rotatebox{270}{\scriptsize #6}}] at (#1+2,#2+1) {};
	 \node[label={[label distance=-1.55mm]above:\rotatebox{0}{\scriptsize #7}}] at (#1+1,#2) {};
	 \node[label={[label distance=-1.55mm]right:\rotatebox{270}{\scriptsize #8}}] at (#1,#2+1) {}; 
}

\usepackage{graphicx} 

\begin{document}

\title{Compact Error-Resilient Self-Assembly of Recursively Defined Patterns}
\titlerunning{} 

\author{Brad Shutters\inst{1} \and
Timothy P. Hartke Jr.\inst{1} \and
Robert J. Sammelson\inst{2}}

\authorrunning{B. Shutters et al.}

\institute{Ball State University, Muncie, IN 47306, USA \email{\{bsshutters,tphartke\}@bsu.edu} \and
Burris Laboratory School, Muncie, IN 47306, USA \email{rjsammelson@bsu.edu}}

\maketitle            
\begin{abstract}
A limitation to molecular implementations of tile-based self-assembly systems is the high rate of mismatch errors which has been observed to be between 1\% and 10\%.
Controlling the physical conditions of the system to reduce this intrinsic error rate $\epsilon$ prohibitively slows the growth rate of the system. 
This has motivated the development of techniques to redundantly encode information in the tiles of a system in such a way that the rate of mismatch errors in the final assembly is reduced even without a reduction in $\epsilon$. 
Winfree and Bekbolatov, and Chen and Goel, introduced such error-resilient systems that reduce the mismatch error rate to $\epsilon^k$ by replacing each tile in an error-prone system with a $k \times k$ block of tiles in the error-resilient system, but this increases the number of tile types used by a factor of $k^2$, and the scale of the pattern produced by a factor of $k$. 
Reif, Sahu and Yin, and Sahu and Reif, introduced compact error-resilient systems for the self-assembly of Boolean arrays that reduce the mismatch error rate to $\epsilon^2$ without increasing the scale of the pattern produced. 
In this paper, we give a technique to design compact error-resilient systems for the self-assembly of the recursively defined patterns introduced by Kautz and Lathrop. We show that our compact error-resilient systems reduce the mismatch error rate to $\epsilon^2$ by using the independent error model introduced by Sahu and Reif. Surprisingly, our error-resilient systems use the same number of tile types as the error-prone system from which they are constructed. 
\end{abstract}
%

\section{Introduction}

In an algorithmic self-assembly of DNA molecules, information is encoded into the sticky ends of the molecules in such a way that when the molecules are placed in solution, they self-assemble into some target pattern, structure, or computation \cite{winfree1999universal}. Although implementing this process in laboratory experiments is technically challenging \cite{rothemund2004algorithmic}, it has many applications to nanotechnology \cite{seeman2003dna}. 

A standard computational model for DNA self-assembly is the {\em abstract Tile Assembly Model} (aTAM) introduced in \cite{rothemund2000program}. 
In this model, a DNA molecule is represented by a {\em tile} which is a unit square that can be translated, but not rotated, so that each of its four sides are well-defined. On each side of a tile is a {\em glue} that is made up of both a non-negative integer {\em strength} (usually 0, 1, or 2), and a string (over some finite alphabet) which is called the glue's {\em color}. Each glue represents the information encoded into a sticky end of the molecule.
We say that two tiles are of the same {\em type} if they have the identical glues on each of their corresponding sides (in both color and strength). Addtionally, each tile $t$ has a {\em label}, denoted by $\tlabel(t)$, that is associated with its tile type, but does not play a role in the assembly process.
Several tile types are illustrated in Fig. \ref{fig:TileSet}. 

\begin{figure}[!t]
\centering
\subfloat[][A set of tile types.] {\label{fig:TileSet}
	\hspace{1em}
	\begin{tikzpicture}[scale=.41]
		
		 \draw[rounded corners] (-.35,6.45) rectangle (4.9,12.4);
		 \node at (2.2,11.8) {{\footnotesize Interior}};
		 
		 \tile{0}{9.1}{\shadecolor}{1}{1}{1}{1}{0}
		 \tile{2.5}{9.1}{\shadecolor}{1}{1}{1}{0}{1}	
		 
		 \tile{0}{6.8}{white}{0}{0}{0}{1}{1}
		 \tile{2.5}{6.8}{white}{0}{0}{0}{0}{0}
	
		 \draw[rounded corners] (-.35,0) rectangle (4.9,6.1);
		 \node at (2.2,5.4) {{\footnotesize Boundary}};
		 
		 \vtile{0}{2.75}{\shadecolor}{1}{1}{1}{1}{}	 
		 \htile{2.5}{2.75}{\shadecolor}{1}{1}{1}{}{1}
		 
		 \node at (1,1.35) {$\sigma=$};
		 \seed{1.75}{.35}{\shadecolor}{1}{1}{1}{}{}
		
		 \node at (0,0) {};
	\end{tikzpicture}
	\hspace{1em}
}
\hspace{1em}
\subfloat[][The self-assembly process.] {\label{fig:AssemblyProcess}
	\hspace{1em}
	\begin{tikzpicture}[scale=.41]
		
		\seed{0}{0}{\shadecolor}{$1$}{1}{1}{}{}
		
		\foreach \i in {1,...,5}
		{
  			 \vtile{0}{2*\i}{\shadecolor}{1}{1}{1}{1}{}	
		}
		
		\foreach \i in {1,...,6}
		{					 
			\htile{2*\i}{0}{\shadecolor}{1}{1}{1}{}{1}	
		}
		
		\tile{2}{2}{white}{0}{0}{0}{1}{1}
		\tile{2}{4}{\shadecolor}{1}{1}{1}{0}{1}
		\tile{2}{6}{white}{0}{0}{0}{1}{1}
		\tile{2}{8}{\shadecolor}{1}{1}{1}{0}{1}
		\tile{2}{10}{white}{0}{0}{0}{1}{1}
		
		\tile{4}{2}{\shadecolor}{1}{1}{1}{1}{0}
		\tile{6}{2}{white}{0}{0}{0}{1}{1}
		\tile{8}{2}{\shadecolor}{1}{1}{1}{1}{0}
		\tile{10}{2}{white}{0}{0}{0}{1}{1}
		\tile{12}{2}{\shadecolor}{1}{1}{1}{1}{0}
		
		\tile{4}{4}{white}{0}{0}{0}{1}{1}
		\tile{6}{4}{white}{0}{0}{0}{0}{0}
		\tile{4}{6}{white}{0}{0}{0}{0}{0}
		\tile{6}{6}{white}{0}{0}{0}{0}{0}
		
		\tile{8}{4}{\shadecolor}{1}{1}{1}{1}{0}
		\tile{8}{6}{\shadecolor}{1}{1}{1}{1}{0}
		
		\tile{4}{8}{\shadecolor}{1}{1}{1}{0}{1}
		
		\tile{10}{4}{\shadecolor}{1}{1}{1}{0}{1}
		
		
		\tile{10}{6}{white}{0}{0}{0}{1}{1}
		
		\tile{12}{4}{white}{0}{0}{0}{1}{1}
		\tile{12}{6}{white}{0}{0}{0}{0}{0}

		 \draw[rounded corners] (4.2,11.7) -- (4.2,12.5) -- (14.2, 12.5) -- (14.2,9.2) 
		 -- (6.3, 9.2) -- (6.3,10.2) -- (4.2,10.2) -- (4.2,11.7) {};
		  
		 \tile{4.35}{10.35}{white}{0}{0}{0}{0}{0}
		 
		 \node at (10, 12) {{\footnotesize This tile can't bind}};
		 \node at (10, 11.25) {{\footnotesize  here since its inter-}};
		 \node at (10.2,10.45) {{\footnotesize action strength with}};
		 \node at (10.35,9.68)  {{\footnotesize the assembly is $1\!<\!\tau$.}};

	\end{tikzpicture}
	\hspace{1em}
}
\caption{A rectilinear TAS. Glue strengths are indicated by dots (strength 0), thin lines (strength 1), or thick lines (strength 2). A tile's label is shown at its center, and tiles with the same label have the same background shade.}
\label{fig:aTAM}
\end{figure}
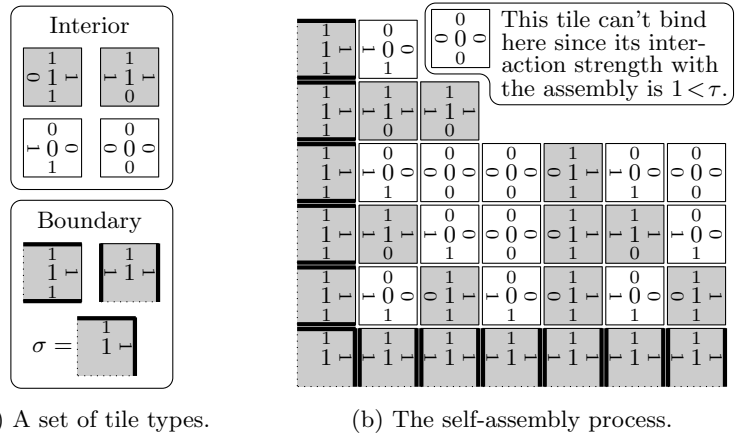

A {\em tile assembly system} (TAS) in the aTAM is a triple $(T,\sigma,\tau)$ where $T$ is a finite set of tile types, $\sigma \in T$ is the type of the {\em seed} tile, and the {\em temperature} $\tau$ is some positive integer which, in this paper, is always 2.
There is an unbounded number of tiles of each type in $T$ which self-assemble into an aggregate {\em assembly}, represented as a mapping from $\Z^2$ to $T\,\dot{\cup}\,\{\bot\}$, by tiles binding, one at a time, to an initial assembly that consists of a single tile of type $\sigma$ placed at the origin.
Whether a tile can bind to an assembly at a position $(x,y)$ is determined by how it interacts with the tiles at adjacent locations in the assembly. 
Two tiles {\em interact} precisely when they are located next to each other, and have the same glue (in both color and strength) on their abutting side.
The {\em interaction strength} of a tile with an assembly at a position $(x,y)$ is the summed strength of the glues on sides of the tile that interact with tiles at positions adjacent to $(x,y)$ in the assembly. 
A tile can {\em bind} to the assembly at a position $(x,y)$ if its interaction strength with the assembly at $(x,y)$ is at least the temperature $\tau$.
See Fig. \ref{fig:AssemblyProcess} for an illustration of the self-assembly process for a TAS using the tile set in Fig. \ref{fig:TileSet}.
When no more tiles can bind to the assembly it is said to be a {\em final assembly} of $\T$.
All TASs considered in this paper result in a unique final assembly (which is allowed to be infinite) in the aTAM.\footnote{It is a straightforward exercise to show that all of the TASs considered here result in a unique final assembly by using the local determinism technique given in \cite{soloveichik2007complexity}.}
We write $\T[x,y]$ for the tile at position $(x,y)$ in the unique final assembly of the TAS $\T$, or $\bot$ if no such tile exists.

A TAS $\T=(T,\sigma,2)$ is {\em rectilinear} if there is a partition $\{B, I\}$ of $T$ 
where tiles of types in $B$ are used for the boundary of the assembly, and tiles of types in $I$ have only strength 1 glues and are used for the interior of the assembly. A rectilinear TAS has the property that, for every $(x,y) \in \Z^2$ where $x,y\!>\!0$, the tile $\T[x,y]$ binds to the assembly by interacting with the tiles $\T[x-1,y]$ and $T[x,y-1]$. This forces the assembly to grow in the northeast direction. 
For a rectiliner TAS $\T$, we denote by $L_\T(x,y)$ the value of $\tlabel(\T[x,y])$, we denote by $H_\T(x,y)$ the color of the glue shared by $\T[x-1,y]$ and $\T[x,y]$ on their abutting side, and we denote by $V_\T(x,y)$ the color of the glue shared by the tiles $\T[x,y-1]$ and $\T[x,y]$ on their abutting side. 
When $\T$ is clear, we write $L(x,y)$, $H(x,y)$ and $V(x,y)$ for $L_\T(x,y)$, $H_\T(x,y)$ and $V_\T(x,y)$, respectively.
For example, the TAS illustrated in Fig. \ref{fig:aTAM} is a rectilinear TAS where $V(x,y)=L(x,y-1)$, $H(x,y)=L(x-1,y)$, and $L(x,y)= (V(x,y) + H(x,y))\!\!\!\mod 2$.

We are interested in the pattern produced by a rectilinear TAS  $\T=(T,\sigma,2)$. A {\em pattern} of some finite alphabet $\L$ is a mapping from $\Z^2$ to $\L\,\dot{\cup}\,\{\bot\}$. The {\em pattern produced by a TAS $\T$}, denoted by $P(\T)$, is the pattern of $\cup_{t \in T} \tlabel(t)$ that maps each $(x,y) \in \Z^2$ to $\tlabel(\T[x,y])$, or $\bot$ when $\T[x,y]=\bot$.
We say that a pattern $P$ {\em self-assembles} if there is a TAS $\T$ such that $P=P(\T)$.
For example, the TAS of Fig. \ref{fig:aTAM} produces the Sierpinski triangle pattern (see Fig. \ref{fig:SierpinskiTriangle}). 
Several complex patterns have been shown to self-assemble in the aTAM \cite{rothemund2004algorithmic,barish2005two,barish2009information,fujibayashi2008toward,kautz2009self}.

\begin{figure}[!b]
\centering
\subfloat[][Initial mismatch error.]{
	\hspace{1em}
	\begin{tikzpicture}[scale=.41]
		\seed{0}{0}{\shadecolor}{$1$}{1}{1}{}{}
		
		\foreach \i in {1,...,3}
		{
  			 \vtile{0}{2*\i}{\shadecolor}{$1$}{1}{1}{1}{}	
			 \htile{2*\i}{0}{\shadecolor}{$1$}{1}{1}{1}{}	
		}
		
		\tile{2}{2}{white}{0}{0}{0}{1}{1}
		\tile{2}{4}{\shadecolor}{1}{1}{1}{0}{1}
		\tile{2}{6}{white}{0}{0}{0}{1}{1}
		
		\tile{4}{2}{\shadecolor}{1}{1}{1}{1}{0}
		\tile{6}{2}{white}{0}{0}{0}{1}{1}
		
		\tile{4}{4}{\shadecolor}{1}{1}{1}{0}{1}
		
		\draw[red,ultra thick] (5,4) ellipse (3.5mm and 7mm);
	\end{tikzpicture}
	\hspace{1em}
}
\hspace{2em}
\subfloat[][Permanent mismatch error.]{
	\hspace{1em}
	\begin{tikzpicture}[scale=.41]
		\seed{0}{0}{\shadecolor}{$1$}{1}{1}{}{}
		
		\foreach \i in {1,...,3}
		{
  			 \vtile{0}{2*\i}{\shadecolor}{$1$}{1}{1}{1}{}	
			 \htile{2*\i}{0}{\shadecolor}{$1$}{1}{1}{1}{}	
		}
		
		\tile{2}{2}{white}{0}{0}{0}{1}{1}
		\tile{2}{4}{\shadecolor}{1}{1}{1}{0}{1}
		\tile{2}{6}{white}{0}{0}{0}{1}{1}
		
		\tile{4}{2}{\shadecolor}{1}{1}{1}{1}{0}
		\tile{6}{2}{white}{0}{0}{0}{1}{1}	
		
		\tile{4}{4}{\shadecolor}{1}{1}{1}{0}{1}
		
		\draw[red,ultra thick] (5,4) ellipse (3.5mm and 7mm);
		
		\tile{4}{6}{\shadecolor}{1}{1}{1}{1}{0}
		\tile{6}{4}{\shadecolor}{1}{1}{1}{0}{1}
		\tile{6}{6}{white}{0}{0}{0}{1}{1}
	\end{tikzpicture}
	\hspace{1em}
}
\caption{Mismatch errors.}
\label{fig:MismatchErrors}
\end{figure}
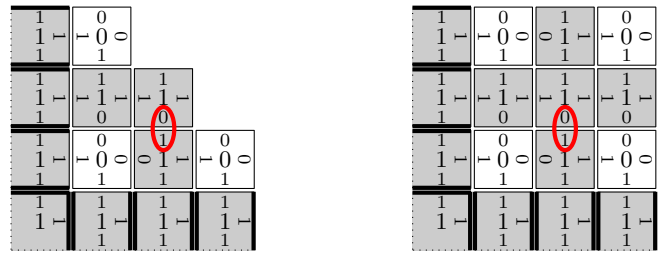

Although the aTAM is an error-free model of self-assembly, a limitation encountered in molecular implementations of rectilinear TASs \cite{rothemund2004algorithmic,winfree2003proofreading} is that sometimes a molecule will initially bind to an assembly with interaction strength less than $\tau$, but before it can detach, other molecules bind to the assembly in such a way that the initially weakly bound molecule becomes permanently ``locked'' into the assembly \cite{rothemund2004algorithmic,winfree2003proofreading}. 
We say that there is a {\em mismatch error} between two adjacent tiles in an assembly to mean that they do not interact, i.e., they have different glues on their abutting sides. 
See Fig. \ref{fig:MismatchErrors} for an illustration of a mismatch error. 
In implementations of rectiliner TASs using DNA double-crossover molecules, the intrinsic rate of mismatch errors $\epsilon$ is observed to range from 1\% to 10\% \cite{rothemund2004algorithmic,fujibayashi2008toward,barish2009information}. 
Not only do mismatch errors allow tiles to bind at incorrect locations in an assembly, but they propagate in the sense that an initial mismatch error invites more tiles to bind at incorrect locations even without a further mismatch error, thereby preventing the reliable self-assembly of large target structures with minimal or no defects.
Although the physical conditions of the system can be optimized to reduce $\epsilon$, doing so slows the rate at which tiles bind to the assembly to prohibitive levels for most applications.

The approach to reducing mismatch errors considered here is to design  {\em error-resilient} TASs that reduce the rate of mismatch errors in their final assembly without requiring a reduction in the intrinsic error rate $\epsilon$ of the system. 
An {\em error-resilient tiling scheme} transforms a TAS $\T$ into an error-resilient TAS $\R$ such that $P(\R)=P(\T)$.  
Such schemes work by redundantly encoding information in the tiles in such a way that mismatch errors become prone to detach \cite{winfree2003proofreading,chen2004error,reif2006compact,sahu2010capabilities}. 
In \cite{schulman2015inc}, it is shown that increasing this redundancy can lead to an exponential decrease in the rate of mismatch errors in the final assembly. 
However, several of these error-resilient tiling schemes \cite{winfree2003proofreading,chen2004error} replace each tile in the original TAS by a $k \times k$ block of tiles in the error-resilient TAS. This increases the number of tile types (and tiles) used by a factor of $k^2$ in the general case, and also increases the scale of the pattern produced by the error-resilient system by a factor of $k$ which is undesirable for many applications. An error-resilient tiling scheme is {\em compact} if it results in a TAS that does not increase the scale of the pattern produced. 

Compact error-resilient tiling schemes were first introduced in \cite{reif2006compact,sahu2010capabilities} where they were used to design error-resilient TASs that produce Boolean array patterns.
A {\em Boolean array} TAS is a rectilinear TAS where the tiles are labeled by bits, and there are two binary Boolean operators $op1$ and $op2$ so that $L(x,y)=V(x,y+1)=V(x,y)\ op1\ H(x,y)$, and $H(x+1,y)=V(x,y)\ op2\ H(x,y)$. Initial bits are given for $L(x,y)$, $H(x,y)$ and $V(x,y)$ when $x=0$ or $y=0$. There are 256 such Boolean array TASs, including those that produce the binary counter and the Sierpinski triangle patterns. Each such TAS uses four distinct tile types for the interior of the assembly. The error-resilient TAS produced by the compact error-resilient tiling scheme of \cite{sahu2010capabilities} reduces the rate of mismatch errors to $\epsilon^2$, but uses up to four times as many tile types as the original error-prone TAS in the general case. An additional scheme is given in \cite{sahu2010capabilities} for the case when $op1$ and $op2$ are pairwise input sensitive in which the error-rate is reduced to $\epsilon^3$. Specific error-resilient TASs for the binary counter and the Sierpinski triangle patterns are given in \cite{reif2006compact}. 

In this paper, we introduce a compact error resilient tiling scheme for recursively defined patterns. A pattern $P$ is {\em recursively defined} if, for each $(x,y) \in \Z^2$, $P(x,y)$ is a finite function of all of the other values of $P$ in some fixed-sized rectangle of $\Z^2$ whose upper right corner is at $(x,y)$. The self-assembly of recursively defined patterns was introduced in \cite{kautz2009self}. Such patterns include, as with Boolean arrays, the Sierpinski triangle. However, since the function used in a recursively defined pattern is not required to be a Boolean operator, the recursively defined patterns also include the Sierpinski carpet and the other numerically self-similar fractal patterns, as well as an infinite number of other patterns. 
A uniform construction is given in \cite{kautz2009self} that takes a recursively defined pattern $P$, and creates a rectilinear TAS $\T$ such that $P(\T)=P$.  
Our main result in this paper is a compact error-resilient tiling scheme that transforms a TAS $\T$ created by the construction of \cite{kautz2009self} into an error-resilient TAS $\R$ such that $P(\R)=P(\T)$. We show that the rate of mismatch errors in the error-resilient TAS created by our construction is reduced to $\epsilon^2$ by using the independent error model introduced in \cite{sahu2010capabilities}. Surprisingly, $\R$ uses the same number of tile types as $\T$. To our knowledge, this is the first example of an error-resilient tiling scheme that does not increase the number of tile types used.

The rest of this paper proceeds as follows: In Section \ref{sec:Preliminaries}, we introduce some general notation and terminology used in the rest of the paper, including some operations on tuples used to simplify several of the notations in our constructions and proofs. In Section \ref{sec:TAM}, we introduce some addtional notations with regards to the aTAM and briefly introduce the kinetic Tile Assembly Model (kTAM). In Section \ref{sec:RecursivelyDefinedPatterns}, we review the recursively defined patterns introduced in \cite{kautz2009self} along with the construction given to create a TAS $\T$ in which a given recursively defined pattern self-assembles. 
In Section \ref{sec:Results}, we give our main result: a compact error-resilient tiling scheme that transforms a recursively defined TAS created by the construction of \cite{kautz2009self} into an error-resilient TAS with the same number of tile types, and producing the same pattern, as the original error-prone TAS, but reduces the rate of mismatch errors to $\epsilon^2$.  In Section \ref{sec:Simulations}, we summarize the results of simulating our error-resilient TASs using the Xgrow software on a few recursively defined test patterns to further verify our results. We conclude in Section \ref{sec:Discussion} with a discussion of our results and some ideas for further research.
%
\section{Preliminaries}
\label{sec:Preliminaries}



We work in the discrete Euclidean plane $\Z^{2}$. We write $U_2$ for the set of all unit vectors in $\Z^{2}$, and we refer to the elements of $U_2$ by by the cardinal directions: $\North=(0,1)$, $\East=(1,0)$, $\South=(0,-1)$, and $\West=(-1,0)$. The {\em neighborhood} of a position $(x,y)$ in $\Z^2$, denoted by $N(x,y)$, is the set of positions that differ from $(x,y)$ by at most 1 in each direction, i.e., the set of all $(x',y') \in \Z^2$ such that $|x'-x| \le 1$ and $|y'-y|\le 1$. 
We write $\Sigma^*$ for the set of all strings over some finite alphabet $\Sigma$. 
For a positive integer $n$, $[n] = \{1, 2, \ldots n\}$.
For an $n$-tuple $t$ and $1 \le i \le n$, we write $t_i$ for the $i$th element of $t$. 
We also write $\first(t)$ for $t_1$ and $\second(t)$ for $t_2$.
When $t$ is a tuple-of-tuples, we write $t_{i,j}$ for the $j$th element of the $i$th element of $t$. 
We define the following operations on tuples in order to simplify some notation in the constructions and proofs of this paper. 
\begin{itemize}\itemsep .5em
\item  The {\em concatenation operation}:  $t \cdot u = (t_1, \ldots, t_m, u_1, \ldots, u_n)$ where $t$ is an $m$-tuple and $u$ is an $n$-tuple. When $t$ is a tuple-of-tuples, $\Pi\, t = t_1 \cdot \ldots \cdot t_m$. 
\item The {\em shift-insert operation}:  $t \mapsfrom e = (t_2, \ldots t_n, e)$ where $t$ is an $n$-tuple and $e$ is any element (possibly a tuple).
\item The {\em deep-shift-insert} operation $t \dshift u = (t_1 \mapsfrom u_1, \ldots, t_n \mapsfrom u_n)$ where $t$ is an $n$-tuple-of-tuples, and $u$ is an $n$-tuple.
\end{itemize}
Unless the order of evaluation is specified with brackets, all of the above operations on tuples are evaluated in left-to-right order so that $\Pi\,t \cdot u = t_1 \cdot t_2 \cdot \ldots \cdot t_{n} \cdot u$, $t \mapsfrom e \mapsfrom f = (t_3, \ldots t_n, e, f)$, and $t \dshift u \dshift v = (t_1 \mapsfrom u_1 \mapsfrom v_1, \ldots, t_n \mapsfrom u_n \mapsfrom v_n)$. 


\subsection{The Tile Assembly Model}
\label{sec:TAM}

In this paper we use both the abstract and kinetic versions of the Tile Assembly Model, the aTAM and the kTAM. An overview of the aTAM was given in the introduction, but we will use a few additional notations: for a tile $t$ and a side $\vec{u} \in U_2$, we write  $\tstrength_t(\vec{u})$ and $\tcolor_t(\vec{u})$, for the strength and color, respectively, of the glue on side $\vec{u}$ of $t$.
The {\em kinetic Tile Assembly Model} (kTAM) is an augmentation of the aTAM that uses the physical conditions of the system as parameters so that mismatch errors in the assembly process can be modeled. 
The kTAM models the rate at which tiles bind and detach from a growing assembly as a function of two unitless free energies $G_{se}$ and $G_{mc}$. 
The value of $G_{se}$ controls the number of base pair bonds that must be broken for a tile to detach from the assembly.
The value of $G_{mc}$ measures the concentration of tiles in the system and controls the forward rate of the assembly process. 
It is suggested in \cite{winfree2003proofreading} that a value for $G_{mc}$ slightly less than twice the value of $G_{se}$ provides an optimal intrinsic error rate for a TAS with a temperature of $\tau=2$.
We refer the reader to \cite{winfree2003proofreading,reif2006compact} for a more complete overview of the kTAM.

\subsection{Recursively Defined Patterns}\label{sec:RecursivelyDefinedPatterns}

Here, we introduce recursively defined patterns, and note that our recursively defined patterns are equivalent to the recursively defined matrices of \cite{kautz2009self} where it was shown that all recursively defined matrices self-assemble in the aTAM. 
Since we use some different notations and terminology, we reproduce the relevant definitions and results of \cite{kautz2009self} for coherence. 
Additionally, we remark that although equivalent, our definition of recursively defined patterns is more refined in the sense that the entries in a recursively defined pattern are allowed to depend on the entries in a $w \times h$ block of the pattern, where as in \cite{kautz2009self}, each entry in a recursively defined matrix depends on the entries in an $n \times n$ block of the matrix.

We use the following tuples to specify particular parts of a pattern $P$ defined relative to a particular location $(x,y) \in \Z^2$:
\begin{align*}
\row_{P}^{n}(x,y) &= (P(x-n+1,y), P(x-n+2,y), \ldots, P(x,y))\\
\block_{P}^{w,h}(x,y) &= (\row_{P,w}(x,y-h+1), \row_{P,w}(x,y-h+2), \ldots, \row_{P,w}(x,y))
\end{align*}
where $n$, $w$, and $h$ are positive integers.
Intuitively, 
$\row_{P}^{n}(x,y)$ describes the values of $P$ at each of the $n$ positions directly to the left of and including position $(x,y)$, and
$\block_{P}^{w,h}(x,y)$ describes the values of $P$ at each position of the $w \times h$ rectangle of $Z^2$ whose upper right corner is at position $(x,y)$.

For integers $w,h \ge 2$, we say that a pattern $P$ of $\L$\ is {\em $(w,h)$-recursively defined} to mean that there is a function $f : (\L\,\dot{\cup}\,\{\bot\})^{\,wh-1} \rightarrow \L $ such that for all $(x,y) \in \Z^2$, 
\begin{align*}
P(x,y) = f \left(\Pi\, \block_{P}^{w,h-1}(x,y-1) \cdot \row_{P}^{w-1}(x-1,y)\,\right)
\end{align*}
when $x,y \ge 0$, and $\bot$\ otherwise.
Note that every $(w,h)$-recursively defined pattern is also both a $(w+1,h)$ and $(w,h+1)$-recursively defined pattern. A pattern is {\em recursively defined} if it is $(w,h)$-recursively defined for some $w$ and $h$.

A main result of \cite{kautz2009self} is that every recursively defined pattern self-assembles.
\begin{theorem}[\cite{kautz2009self}]\label{thm:RecursivePatternsSelfAssemble}
Let $P$ be a $(w,h)$-recursively defined pattern on $\L$. There exists a TAS $\T=(T,\sigma,2)$ where $|T| \le (|\L|+1)^{mn-1}$ and $P(\T)=P$.
\end{theorem}
Theorem \ref{thm:RecursivePatternsSelfAssemble} is proven in \cite{kautz2009self}  using a uniform construction that creates a TAS $\T = (T,\sigma,2)$ in which a given $(w,h)$-recursively defined pattern $P$ self-assembles. We review this construction here using the notation and terminology given above:

\begin{construction}[\cite{kautz2009self}]\label{con:KautzLathrop}\normalfont
For each $r^1, r^2, \ldots, r^{h-1} \in \L_\bot^w$ and $r^h \in \L_\bot^{w-1}$ such that both of the following conditions hold: 
	\begin{enumerate}
	\item[] $\forall i \in [h], \forall j,k \in [w]$, $(r^i_j=\bot$ and $r^i_k\ne\bot$) implies $j < k$\,, and 
	\item[] $\forall i,j \in [h], \forall k \in [w]$, $(r^i_k=\bot$ and $r^j_k \ne \bot)$ implies $i < j$\,, 
	\end{enumerate}
let $block=(r^1, \ldots, r^{h-1})$, and add a tile type $t$ to $T$ where $\tlabel(t)= f_P( \Pi block \cdot r^h)$, the glue color on each side of $t$ is defined by:
\begin{align*}
&\tcolor_t(\West) = r^h\\
&\tcolor_t(\South) = block\\
&\tcolor_t(\East) =  r^h \mapsfrom \tlabel(t)\\
&\tcolor_t(\North) = block \mapsfrom [r^h \cdot (\tlabel(t))]
\end{align*}
and the glue strength on each side of $t$ is 1 except in the following cases:
\begin{enumerate}\itemsep .9em
	\item[] {\em Case 1:\ }\ $r^{h-1}, r^h \in \{\bot\}^*$. In this case $\tstrength(\North)=2$, $\tstrength(\East)=2$, $\tstrength(\South)=0$, $\tstrength(\West)=0$, and $\sigma= t$.
	\item[] {\em Case 2:\ }\ $r^{h-1} \in \{\bot\}^*, r^{h} \not\in \{\bot\}^*$. In this case $\tstrength(\West)=2$, $\tstrength(\East)=2$, $\tstrength(\North)=1$, $\tstrength(\South)=0$ (i.e., $t$ is used to assemble the horizontal boundary),
	\item[] {\em Case 3:\ }\ $r^{h-1} \not\in \{\bot\}^*, r^{h} \in \{\bot\}^*$. In this case $\tstrength(\North)=2$, $\tstrength(\South)=2$, $\tstrength(\East)=1$, $\tstrength(\West)=0$ (i.e., $t$ is used to assemble the vertical boundary).
	\end{enumerate}
\end{construction}

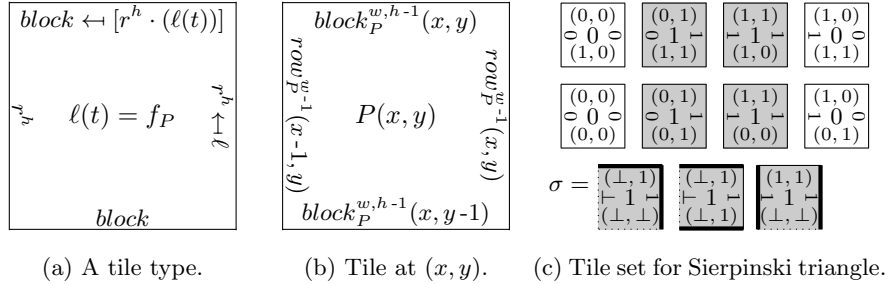
\begin{figure}[!t]
\centering

\subfloat[][A tile type.]{
	\begin{tikzpicture}[scale=1.6]
		\MEDIUMtile{0}{0}{white}{$\tlabel(t)=f_P$}
		{$ block \mapsfrom [r^h \cdot (\tlabel(t))]$}
		{$r^h \shift \ell$}
		{$\block$}
		{$r^h$}
	\end{tikzpicture}
}
\subfloat[][Tile at $(x,y)$.]{
	\begin{tikzpicture}[scale=1.6]
		\MEDIUMtile{0}{0}{white}{$P(x,y)$}
		{$\block_{P}^{w,h\text{\,-}1}(x,y)$}
		{$\row_{P}^{w\text{\,-}1}(x,y)$}
		{$\block_{P}^{w,h\text{\,-}1}(x,y\,\text{-}1)$}
		{$\row_{P}^{w\text{\,-}1}(x\,\text{-}1,y)$}
	\end{tikzpicture}
}
\subfloat[][Tile set for Sierpinski triangle.]{\label{fig:RecursiveTriangleTileSet}
	\begin{tikzpicture}[scale=.45]
	
		\tile{0}{4.8}{white}{0}{$(0,0)$}{$0$}{$(1,0)$}{$0$}
		 \tile{2.4}{4.8}{\shadecolor}{1}{$(0,1)$}{$1$}{$(1,1)$}{$0$}
		 \tile{4.8}{4.8}{\shadecolor}{1}{$(1,1)$}{$1$}{$(1,0)$}{$1$}	
		 \tile{7.2}{4.8}{white}{0}{$(1,0)$}{$0$}{$(1,1)$}{$1$}

		\tile{0}{2.4}{white}{0}{$(0,0)$}{$0$}{$(0,0)$}{$0$}
		 \tile{2.4}{2.4}{\shadecolor}{1}{$(0,1)$}{$1$}{$(0,1)$}{$0$}
		 \tile{4.8}{2.4}{\shadecolor}{1}{$(1,1)$}{$1$}{$(0,0)$}{$1$}	
		 \tile{7.2}{2.4}{white}{0}{$(1,0)$}{$0$}{$(0,1)$}{$1$}
		 
		 \node at (.3,1.35) {$\sigma=$};
		 \seed{1.1}{0}{\shadecolor}{1}{$(\bot,1)$}{$1$}{$(\bot,\bot)$}{$\bot$}
		 
		 \vtile{3.5}{0}{\shadecolor}{1}{$(\bot,1)$}{$1$}{$(\bot,1)$}{$\bot$}	 
		 \htile{5.8}{0}{\shadecolor}{1}{$(1,1)$}{$1$}{$(\bot,\bot)$}{$1$}
		 
		 \node at (0,-.15) {};
	\end{tikzpicture}
}
\caption{Construction \ref{con:KautzLathrop}. (a) shows the design of the glue colors on each side of a tile. (b) shows the information encoded into the sides of the tile that binds at position $(x,y)$ in the assembly. (c) shows the tile set created for the Sierpinski triangle. To save space, all 1-tuples are shown in (c) without parantheses.}
\label{fig:RecursiveTileSet}
\end{figure}

Intuitively, the TAS $\T$ created by Construction \ref{con:KautzLathrop} is a rectilinear TAS where, for each $(x,y) \in \Z^2$, $V(x,y) = \block_{P}^{w,h-1}(x,y-1)$, $H(x,y)=\row_{P}^{w-1}(x-1,y)$, and $L(x,y) = f_P( \Pi\,V(x,y) \cdot H(x,y))$. For example, the TAS created by Construction \ref{con:KautzLathrop} for the Sierpinski triangle has $V(x,y)=(L(x-1,y-1), L(x,y-1))$, $H(x,y)=(L(x-1,y))$, and $L(x,y) = L(x,y-1) + L(x-1,y)\!\!\!\mod 2$. See Fig. \ref{fig:RecursiveTileSet} for an illustration.

\subsubsection{A note on optimizing size of a tile set created by Construction \ref{con:KautzLathrop}.}

For $(2,2)$-recursively defined patterns $P$ in which the value of $P(x,y)$ does not depend on $P(x-1,y-1)$, the size of the tile set in a TAS produced by Construction \ref{con:KautzLathrop} for $P$ is not optimal. 
For example, the tile set for the Sierpinski triangle pattern given in Fig. \ref{fig:RecursiveTriangleTileSet} uses 3 tile types for the boundary and 8 tile types for the interior. However, the tile set for the Sierpinski triangle given in Fig. \ref{fig:aTAM} uses 3 tile types for the boundary and only 4 tile types for the interior. 
Another construction is given in \cite{kautz2009self} that creates a smaller tile set for such patterns. We stress here, that the error-resilient tiling scheme we introduce in Section \ref{sec:Results} requires as input the tile set created exactly as specified by Construction \ref{con:KautzLathrop}.

%
\section{Main Results}
\label{sec:Results}

In this section, we give our main result: a compact error-resilient tiling scheme that transforms a recursively defined TAS created by Construction \ref{con:KautzLathrop} given in \cite{kautz2009self} (and reviewed in Section \ref{sec:RecursivelyDefinedPatterns}) for a recursively defined pattern $P$, into an error-resilient TAS that produces the pattern $P$ and uses the same number of tile types, but the rate of mismatch errors in the final assembly is reduced to $\epsilon^2$. We analyze the error-rate of the TAS created by our error-resilient tiling scheme using the error model introduced in \cite{sahu2010capabilities}.

\begin{theorem}[Main Theorem]\label{thm:MainTheorem}
There exists a compact error-resilient tiling scheme that, given a TAS $\T=(T,\sigma_\T,2)$ created by Construction \ref{con:KautzLathrop}, creates a TAS $\R=(R,\sigma_\R,2)$ such that $P(\R)=P(\T)$,  $|R|=|T|$, and the rate of mismatch errors in the final assembly of $\R$ is $\epsilon^2$, where $\epsilon$ is the intrinsic error rate of the system. 
\end{theorem}

Our proof of Theorem \ref{thm:MainTheorem} uses Construction \ref{con:MainConstruction}, given below, which takes as input a TAS $\T$ created by the construction of \cite{kautz2009self} (Construction \ref{con:KautzLathrop} in Section \ref{sec:RecursivelyDefinedPatterns}), and creates such a TAS $\R$. %
We will need one additional part of a pattern $P$ relative to a particular location $(x,y) \in \Z^2$ to use in Construction \ref{con:MainConstruction}, along with the parts $\row_{P}^{n}(x,y)$ and $\block_{P}^{w,h}(x,y)$ defined in Section \ref{sec:RecursivelyDefinedPatterns}:
\begin{equation*}
\col_{P}^{n}(x,y) = (P(x,y-n+1), P(x,y-n+2), \ldots, P(x,y))
\end{equation*}
where $n$ is some positive integer.
Intuitively, 
$\col_{P}^{n}(x,y)$ describes the values of $P$ at each of the $n$ positions directly below and including position $(x,y)$.

\begin{construction}\label{con:MainConstruction}{\normalfont
Let $\T=(T,\sigma_\T,2)$ be the TAS created for a $(w,h)$-recursively defined pattern $P$ by the construction in the proof of Theorem \ref{thm:RecursivePatternsSelfAssemble}. We construct a TAS $\R=(R,\sigma_\R,2)$ as follows:
For each tile type $t \in T$, let
\begin{align*}
b &= \tcolor_t(\South)\\
row &= \tcolor_t(\West)\\
col &= (b_{1,w}, \ldots, b_{h-1,w})\\
block &=  ( (b_{1,1}, \ldots, b_{1,w-1}), (b_{2,1}, \ldots, b_{2,w-1}), \ldots, (b_{h-1,1}, \ldots, b_{h-1,w-1}) )
\end{align*}
and add a tile type $r_t$ to $R$ where $\tlabel(r_t) = \tlabel(t)$, $\tstrength_{r_t}(\vec{u}) = \tstrength_t(\vec{u})$ for each $\vec{u} \in U_2$, and the glues of $r_t$ are set in the following way:
        	\begin{align*}
	\tcolor_{r_t}(\West) &= (block, row)\\
	\tcolor_{r_t}(\South) &= (block, col)\\
	\tcolor_{r_t}(\East) &=  (block \dshift col, row \mapsfrom \tlabel(t))\\
	\tcolor_{r_t}(\North) &= (block \mapsfrom row, col \mapsfrom \tlabel(t))
	\end{align*}
	After adding a tile type $r_t$ to $R$ for each $t \in T$, we set $\sigma_\R=r_{\sigma_\T}$.
}
\end{construction}

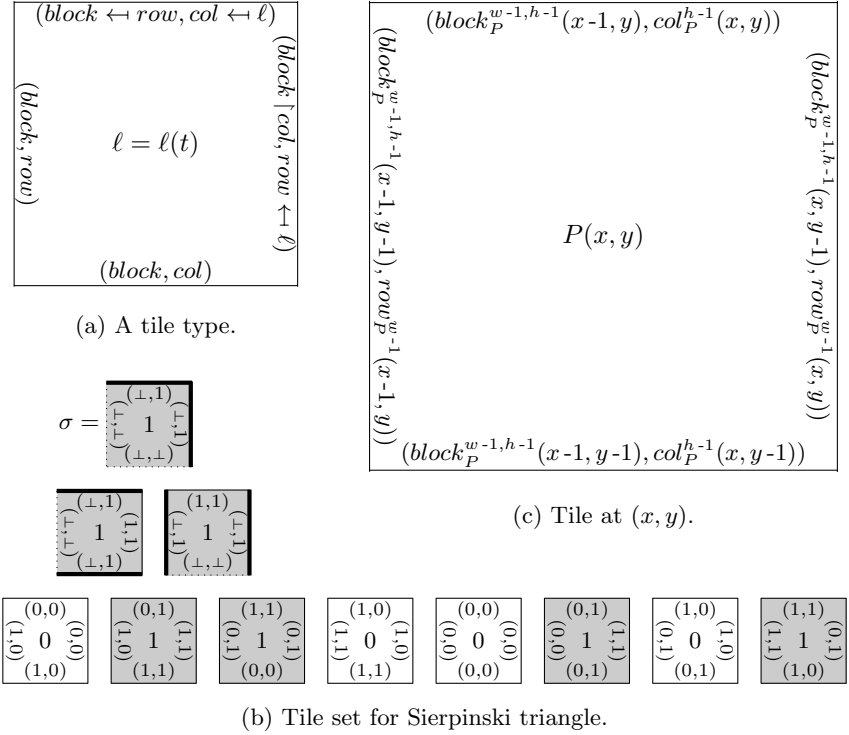
\begin{figure}[!t]

\centering

\begin{minipage}[b]{1.8in}
\centering
\subfloat[][A tile type.]{
	\begin{tikzpicture}[scale=2]
		\MEDIUMtile{0}{0}{white}{$\ell=\ell(t)$}
		{$(block \mapsfrom row, col \mapsfrom \ell)$}
		{$(block \dshift col, row \mapsfrom \ell)$}
		{$(block, col)$}
		{$(block, row)$}
	\end{tikzpicture}
}\\
\renewcommand{\thesubfigure}{b}
\subfloat[][Tile set for  Sierpinski triangle.]{
	\begin{tikzpicture}[scale=.6]

		 \node at (1.75,1.35) {$\sigma=$};
		 \seed{2.3}{.35}{\shadecolor}{1}{$(\text{\tiny$\bot$},\!1)$}{$(\text{\tiny$\bot$},\!1)$}{$(\text{\tiny$\bot$},\!\text{\tiny$\bot$})$}{$(\text{\tiny$\bot$},\!\text{\tiny$\bot$})$}
		 
		 \vtile{1.2}{-2.05}{\shadecolor}{1}{$(\text{\tiny$\bot$},\!1)$}{$(1,\!1)$}{$(\text{\tiny$\bot$},\!1)$}{$(\text{\tiny$\bot$},\!\text{\tiny$\bot$})$}	 
		 \htile{3.6}{-2.05}{\shadecolor}{1}{$(1,\!1)$}{$(\text{\tiny$\bot$},\!1)$}{$(\text{\tiny$\bot$},\!\text{\tiny$\bot$})$}{$(\text{\tiny$\bot$},\!1)$}

		 \tile{0}{-4.45}{white}{0}{$(0,\!0)$}{$(0,\!0)$}{$(1,\!0)$}{$(1,\!0)$}
		 \tile{2.4}{-4.45}{\shadecolor}{1}{$(0,\!1)$}{$(1,\!1)$}{$(1,\!1)$}{$(1,\!0)$}
		 \tile{4.8}{-4.45}{\shadecolor}{1}{$(1,\!1)$}{$(0,\!1)$}{$(0,\!0)$}{$(0,\!1)$}	
		  \tile{7.2}{-4.45}{white}{0}{$(1,\!0)$}{$(1,\!0)$}{$(1,\!1)$}{$(1,\!1)$}

		 \tile{9.6}{-4.45}{white}{0}{$(0,\!0)$}{$(0,\!0)$}{$(0,\!0)$}{$(0,\!0)$}
		 \tile{12}{-4.45}{\shadecolor}{1}{$(0,\!1)$}{$(1,\!1)$}{$(0,\!1)$}{$(0,\!0)$}
		 \tile{14.4}{-4.45}{white}{0}{$(1,\!0)$}{$(1,\!0)$}{$(0,\!1)$}{$(0,\!1)$}
		 \tile{16.8}{-4.45}{\shadecolor}{1}{$(1,\!1)$}{$(0,\!1)$}{$(1,\!0)$}{$(1,\!1)$}
;
	
	\end{tikzpicture}
}
\end{minipage}
\begin{minipage}[b]{2.8in}
\centering
\renewcommand{\thesubfigure}{c}
\subfloat[][Tile at $(x,y)$.]{
	\begin{tikzpicture}[scale=3.3]
		\LARGEtile{0}{0}{white}{$P(x,y)$}
		{$(\block_{P}^{w\text{\,-}1,h\text{\,-}1}(x\text{\,-}1,y), \col_{P}^{h\text{\,-}1}(x,y))$}
		{$(\block_{P}^{w\text{\,-}1,h\text{\,-}1}(x,y\text{\,-}1), \row_{P}^{w\text{\,-}1}(x,y))$}
		{$(\block_{P}^{w\text{\,-}1,h\text{\,-}1}(x\text{\,-}1,y\,\text{-}1), \col_{P}^{h\text{\,-}1}(x,y\text{\,-}1))$}
		{$(\block_{P}^{w\text{\,-}1,h\text{\,-}1}(x\text{\,-}1,y\,\text{-}1), \row_{P}^{w\text{\,-}1}(x\,\text{-}1,y))$}
	\end{tikzpicture}
}
\vspace{6.625em}
\end{minipage}
\caption{Construction \ref{con:MainConstruction}. (a) shows the design of the glue colors on each side of a tile. (b) shows the tile set created for the Sierpinski triangle. To save space, all 1-tuples are shown in (b) without parantheses. (c) shows the information encoded into the sides of tile that binds at position $(x,y)$ in the assembly.}
\label{fig:MainConstruction}
\end{figure}

Intuitively, the TAS $\R$ created by Construction \ref{con:MainConstruction} is a rectilinear TAS in which, for each $(x,y) \in \Z^2$, $V(x,y)$ represents the pair $(\block_{P}^{w-1,h-1}(x-1,y-1), \col_{P}^{h-1}(x,y-1))$, $H(x,y)$ represents the pair $(\row_{P}^{w-1}(x-1,y), \block_{P}^{w-1,h-1}(x-1,y-1)$, and $L(x,y) = f_P( \Pi\,\block_{P}^{w,h-1}(x,y-1) \cdot \row_{P}^{w-1}(x-1,y))$. 
Note that $\block_{P}^{w-1,h-1}(x-1,y-1)$ is encoded redundantly in both $V(x,y)$ and $H(x,y)$.
See Figure \ref{fig:MainConstruction} for an illustration.

The following lemma shows that the TAS $\R$ created by Construction \ref{con:MainConstruction} for a given TAS $\T$ created by Construction \ref{con:KautzLathrop} produces the same pattern as $\T$ and uses the same number of  tile types. 

\begin{lemma}\label{lem:SamePatternAndSize}
Let $\T=(T,\sigma_\T,2)$ be a TAS created by Construction \ref{con:KautzLathrop}, and let $\R=(R,\sigma_\R,2)$ be the TAS created by Construction \ref{con:MainConstruction} given $\T$. Then $P(\R)=P(\T)$ and $|R|=|T|$.
\end{lemma}

The proof of Lemma \ref{lem:SamePatternAndSize} is straightforward and given in the technical appendix. 

\subsection{Error Analysis}
\label{sec:ErrorAnalysis}
We analyze the error rate of a TAS $\R$ created by Construction \ref{con:MainConstruction} by using the independent error model introduced in \cite{sahu2010capabilities}. In this model, $\epsilon$ represents the probability of the event that there is a mismatch error between a pair of tiles, and they stay connected in the equilibrium. It is assumed that $\epsilon$ is independent of any other mismatch errors, or lack thereof. We refer the reader to \cite{sahu2010capabilities} for more on the independent error model.

\begin{theorem}\cite{sahu2010capabilities}\label{thm:IndependentErrorModel} Let $\R$ be a rectilinear TAS, and let $(x,y) \in \Z^2$. Under the independent error model, if a mismatch error between the tiles $\R[x,y]$ and $\R[x,y-1]$ or between the tiles $\R[x,y]$ and $\R[x-1,y]$ forces $k$ further mismatch errors between tiles in $N(x,y)$, then the rate of mismatch errors in the final assembly of $\R$ is reduced to $\epsilon^{k+1}$. 
\end{theorem}

We will show that an initial mismatch error in a TAS $\R$ created by Construction \ref{con:MainConstruction} forces at least one additional mismatch error in the neighborhood of the initial mismatch error, thereby reducing the rate of mismatch errors in the final assembly to $\epsilon^{2}$.  We emphasize here that a mismatch error is between the entire glues on the abutting sides of two tiles. In our case, each glue color represents a pair of tuples $(u,v)$, and so a mismatch between any of the respective elements of the two glue colors on the abutting sides of two tiles constitutes a mismatch error between those two tiles. We will need the following lemma. 

\begin{lemma}\label{lem:NoUVWMismatches}
Let $\R$ be a TAS created by Construction \ref{con:MainConstruction} and let $(x,y) \in \Z^2$ where $x>0$ and $y>0$. Let $u=\T[x-1,y-1]$, $v=\T[x,y-1]$ and $w=\T[x-1,y]$. If there is no mismatch error between $u$ and $v$, and no mismatch error between $u$ and $w$, then all of the following equalities hold:
\begin{align*}
	\first(\gcolor_v(\north)) & \overset{(1)}{=} \first(\gcolor_w(\east))\\	
	\first(\gcolor_v(\north)) \dshift \second(\gcolor_v(\north)) & \overset{(2)}{=} \first(\gcolor_v(\east)) \shift \second(\gcolor_v(\east))\\
	\first(\gcolor_w(\east)) \shift \second(\gcolor_w(\east)) & \overset{(3)}{=} \first(\gcolor_w(\north)) \dshift \second(\gcolor_w(\north))\,.
\end{align*}
\end{lemma}

 The proof of Lemma \ref{lem:NoUVWMismatches} is straightforward and given in the technical appendix. 

\begin{theorem}\label{thm:RForcesOneMore}
Let $\R$ be a TAS created by Construction \ref{con:MainConstruction}, and let $(x,y) \in \Z^2$.
A mismatch error between the tiles $\R[x,y]$ and $\R[x,y-1]$ or between the tiles $\R[x,y]$ and $\R[x-1,y]$, forces a further mismatch error between tiles in $N(x,y)$.
\end{theorem}

\begin{proof}
We refer to the tiles in $N(x,y)$ by $t=\R[x,y]$, $u=\R[x-1,y-1]$, $v=\R[x,y-1]$, $w=\R[x,y-1]$, $p=\R[x+1,y-1]$, $q=\R[x+1,y]$, $r=\R[x-1,y+1]$, $s=\R[x,y+1]$, and $x=\R[x+1,y+1]$. 
By assumption, the initial mismatch error is either between the tiles $t$ and $v$, or between the tiles $t$ and $w$.
We can further assume that there are no mismatch errors between the tiles $u$ and $v$, or between the tiles $u$ and $w$, otherwise the theorem is trivially true. Hence, the equalities (1), (2) and (3) given in Lemma \ref{lem:NoUVWMismatches} hold.
Our goal is to show that, given an initial mismatch error either between the tiles $t$ and $v$ or between the tiles $t$ and $w$, there is a further mismatch error in $N(x,y)$.
We separate into two cases:

{\em Case 1:} The initial mismatch error is between the tiles $t$ and $v$. Then, either $\first(\gcolor_t(\south)) \ne \first(\gcolor_v(\north))$ or $\second(\gcolor_t(\south)) \ne \second(\gcolor_v(\north))$. 

{\em Case 1a:} In this case we assume that the mismatch error between the tiles $t$ and $v$ is because $\first(\gcolor_t(\south)) \ne \first(\gcolor_v(\north))$. 
Then,
\begin{align*}
	\first(\gcolor_t(\west)) &\overset{\text{(a)}}{=} \first(\gcolor_t(\south))\\
					 &\overset{\text{(b)}}{\ne} \first(\gcolor_v(\north))\\
					 &\overset{\text{(c)}}{=} \first(\gcolor_w(\east))
\end{align*}
where (a) follows from Construction \ref{con:MainConstruction}; (b) follows from the assumption of Case 1a; and (c) follows from Lemma \ref{lem:NoUVWMismatches}(1).
So we have shown that in Case 1a, $\first(\gcolor_t(\west)) \ne \first(\gcolor_w(\east))$, and hence, there is a further mismatch error between the tiles $t$ and $w$. 

{\em Case 1b:} In this case we assume that Case 1a does not hold, and so the mismatch error between the tiles $t$ and $v$ is because $\second(\gcolor_t(\south)) \ne \second(\gcolor_v(\north))$.
For sake of contradiction, suppose that there are no mismatch errors between the tiles $t$ and $q$, between the tiles $q$ and $p$, and between the tiles $p$ and $v$. 
Then, 
\begin{align*}
	\first(\gcolor_t(\east)) &\overset{\text{(d)}}{=} \first(\gcolor_t(\south)) \dshift \second(\gcolor_t(\south))\\
					 &\overset{\text{(e)}}{\ne} \first(\gcolor_v(\north))  \dshift \second(\gcolor_v(\north))\\
					 &\overset{\text{(f)}}{=} \first(\gcolor_v(\east)) \shift \second(\gcolor_v(\east))\\
					 &\overset{\text{(g)}}{=} \first(\gcolor_p(\west)) \shift \second(\gcolor_p(\west)))\\
					 &\overset{\text{(h)}}{=} \first(\gcolor_p(\north))\\
					 &\overset{\text{(i)}}{=} \first(\gcolor_q(\south))\\
					 &\overset{\text{(j)}}{=} \first(\gcolor_q(\west))\\
					 &\overset{\text{(k)}}{=} \first(\gcolor_t(\east))
\end{align*}
where (d), (h) and (j) follow from Construction \ref{con:MainConstruction}; (e) follows from the assumption of Case 1b; (f) follows from Lemma \ref{lem:NoUVWMismatches}(2); (g) follows from the assumption that there is no mismatch error between $p$ and $v$; (i) follows from the assumption that there is no mismatch error between $q$ and $p$; and (k) follows from the assumption that there is no mismatch error between $t$ and $q$. Thus, under the supposition that there are no mismatch errors between $t$ and $q$, between $q$ and $p$, and between $p$ and $v$, we have that $\first(\gcolor_t(\east)) \ne \first(\gcolor_t(\east))$ which is a contradiction. Hence, there must be a further mismatch error either between the tiles $t$ and $q$, or between the tiles $q$ and $p$, or between the tiles $p$ and $v$. 

\begin{figure}[!t]
\centering
\subfloat[][Case 1]{

	\begin{tikzpicture}[scale=.5]

	\draw (0,0) rectangle (4,2);
	\draw (0,0) rectangle (2,4);
	

	 \fill[red] (3,2) circle (6pt);
	 \draw[black] (3,2) circle (6pt);
	 
          \fill[blue] (2,3) circle (6pt);
	  \draw[black] (2,3) circle (6pt);
	  
	  \fill[green] (4,3) circle (6pt);
	  \draw[black] (4,3) circle (6pt);
	  
	  \fill[green] (5,2) circle (6pt);
	  \draw[black] (5,2) circle (6pt);
	  
	  \fill[green] (4,1) circle (6pt);
	  \draw[black] (4,1) circle (6pt);

	\draw[step=2cm] (0,0) grid (6,6);
	
	 \node at (1,1) {$\tlabel(u)$} ;
	 \node at (3,1) {$\tlabel(v)$} ;
	 \node at (5,1) {$\tlabel(p)$} ;
	 
	 \node at (1,3) {$\tlabel(w)$} ;
	 \node at (3,3) {$\tlabel(t)$} ;
	 \node at (5,3) {$\tlabel(q)$} ;
	 
	 \node at (1,5) {$\tlabel(r)$} ;
	 \node at (3,5) {$\tlabel(s)$} ;
	 \node at (5,5) {$\tlabel(x)$} ;

	\end{tikzpicture}
	\label{fig:NeighborhoodCase1}
}
\hspace{3em}
\subfloat[][Case 2]{

	\begin{tikzpicture}[scale=.5]

	\draw (0,0) rectangle (4,2);
	\draw (0,0) rectangle (2,4);

          \fill[red] (2,3) circle (6pt);
	 \draw[black] (2,3) circle (6pt);
	  
	 \fill[blue] (3,2) circle (6pt);
	 \draw[black] (3,2) circle (6pt);
	 
	  \fill[green] (3,4) circle (6pt);
	  \draw[black] (3,4) circle (6pt);
	  
	  \fill[green] (2,5) circle (6pt);
	  \draw[black!75] (2,5) circle (6pt);
	  
	  \fill[green] (1,4) circle (6pt);
	  \draw[black] (1,4) circle (6pt);
	 
	\draw[step=2cm] (0,0) grid (6,6);
	
	 \node at (1,1) {$\tlabel(u)$} ;
	 \node at (3,1) {$\tlabel(v)$} ;
	 \node at (5,1) {$\tlabel(p)$} ;
	 
	 \node at (1,3) {$\tlabel(w)$} ;
	 \node at (3,3) {$\tlabel(t)$} ;
	 \node at (5,3) {$\tlabel(q)$} ;
	 
	 \node at (1,5) {$\tlabel(r)$} ;
	 \node at (3,5) {$\tlabel(s)$} ;
	 \node at (5,5) {$\tlabel(x)$} ;

	\end{tikzpicture}
	\label{fig:NeighborhoodCase2}
}

\caption{The neighborhood of a mismatch error. (a) shows Case 1 in the proof of Theorem \ref{thm:RForcesOneMore} where the initial mismatch error at the red location forces a further mismatch at the blue location (in Case 1a), or at one of the green locations (in Case 1b). (b) shows Case 2 in the proof of Theorem \ref{thm:RForcesOneMore} where the initial mismatch error at the red location forces a further mismatch at the blue location (in Case 2a), or at one of the green locations (in Case 2b).}
\label{fig:Neighborhood}
\end{figure}
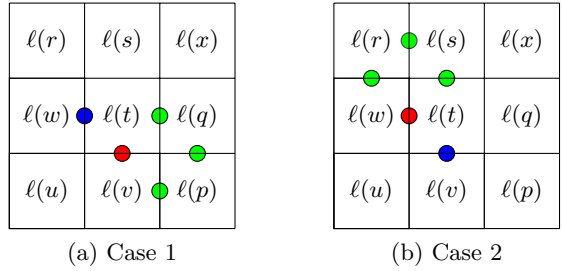

{\em Case 2:} The initial mismatch error is between the tiles $t$ and $w$. Then, either $\first(\gcolor_t(\west)) \ne \first(\gcolor_w(\east))$ or $\second(\gcolor_t(\west)) \ne \second(\gcolor_w(\east))$. 

{\em Case 2a:} In this case we assume that the mismatch error between the tiles $t$ and $w$ is because $\first(\gcolor_t(\west)) \ne \first(\gcolor_w(\east))$. 
Then,
\begin{align*}
	\first(\gcolor_t(\south)) &\overset{\text{(l)}}{=} \first(\gcolor_t(\west))\\
					 &\overset{\text{(m)}}{\ne} \first(\gcolor_w(\east))\\
					 &\overset{\text{(n)}}{=} \first(\gcolor_v(\north))
\end{align*}
where (l) follows from Construction \ref{con:MainConstruction}; (m) follows from the assumption of Case 2a; and (n) follows from Lemma \ref{lem:NoUVWMismatches}(1).
So we have shown that in Case 2a, $\first(\gcolor_t(\south)) \ne \first(\gcolor_v(\north))$, and hence, there is a further mismatch error between the tiles $t$ and $v$. 

{\em Case 2b:}  In this case we assume that Case 2a does not hold, and so the mismatch error between the tiles $t$ and $w$ is because $\second(\gcolor_t(\west)) \ne \second(\gcolor_w(\east))$. 
For sake of contradiction, suppose that there are no mismatch errors between the tiles $t$ and $s$, between the tiles $s$ and $r$, and between the tiles $r$ and $w$. 
Then, 
\begin{align*}
	\first(\gcolor_t(\north)) &\overset{\text{(o)}}{=} \first(\gcolor_t(\west)) \shift \second(\gcolor_t(\west))\\
					  &\overset{\text{(p)}}{\ne} \first(\gcolor_w(\east)) \shift \second(\gcolor_w(\east))\\
					  &\overset{\text{(q)}}{=} \first(\gcolor_w(\north)) \dshift \second(\gcolor_w(\north))\\
					  &\overset{\text{(r)}}{=} \first(\gcolor_r(\south)) \dshift \second(\gcolor_r(\south)))\\
					  &\overset{\text{(s)}}{=} \first(\gcolor_r(\east))\\
					  &\overset{\text{(t)}}{=} \first(\gcolor_s(\west))\\
					  &\overset{\text{(u)}}{=} \first(\gcolor_s(\south))\\
					  &\overset{\text{(v)}}{=} \first(\gcolor_t(\north))
\end{align*}
where (o), (s) and (u) follow from Construction \ref{con:MainConstruction}; (p) follows from the assumption of Case 2b; (q) follows from Lemma \ref{lem:NoUVWMismatches}(3); (r) follows from the assumption that there is no mismatch error between $r$ and $w$; (f) follows from the assumption that there is no mismatch error between $s$ and $r$; and (v) follows from the assumption that there is no mismatch error between $t$ and $s$. 
 Thus, under the supposition that there are no mismatch errors between $t$ and $s$, between $s$ and $r$, and between $r$ and $w$, we have that $\first(\gcolor_t(\north))\ne \first(\gcolor_t(\north))$ which is a contradiction. Hence, there must be a further mismatch error either between the tiles $t$ and $s$, or between the tiles $s$ and $r$, or between the tiles $r$ and $w$. 

In all cases, there is a further mismatch error between tiles in $N(x,y)$. 
\qed
\end{proof}

See Figure \ref{fig:Neighborhood} for an illustration of the proof of Theorem \ref{thm:RForcesOneMore}. Theorems \ref{thm:IndependentErrorModel} and \ref{thm:RForcesOneMore} immediately give the following corollary.

\begin{corollary}\label{cor:RPEps2}
 The rate of mismatch errors in the final assembly of a TAS $\R$ created by Construction \ref{con:MainConstruction} is $\epsilon^2$.
\end{corollary}

\section{Simulations}
\label{sec:Simulations}

\begin{figure}[!b]
\centering

\subfloat[][Sierpinski Triangle $\S$]{\label{fig:SierpinskiTriangle}
	\includegraphics[trim={0 0 257px 257px}, clip, scale=0.36]{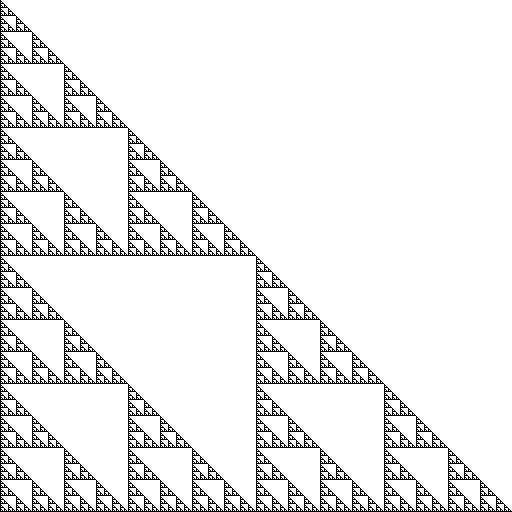}
}
\hspace{1em}
\subfloat[][Sierpinski Carpet $\C$]{\label{fig:SierpinskiCarpet}
	\includegraphics[trim={0 0 269px 269px}, clip, scale=0.3793]{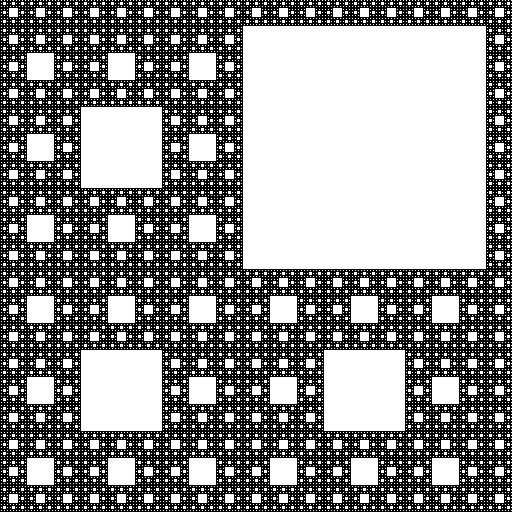}
}
\hspace{1em}
\subfloat[][Pattern $\W$]{\label{fig:PatternW}
	\includegraphics[trim={0  0 257px 257px}, clip, scale=0.36]{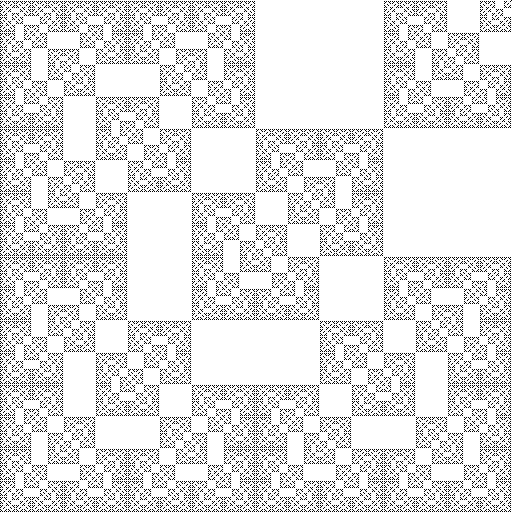}
}
\caption{Some recursively defined patterns.}
\label{fig:NumericallySelfSimilarFractals}
\end{figure}

In this section, we summarize the results of computer simulations we performed using the Xgrow simulator \cite{winfree2003proofreading} to further verify the performance of the compact error-resilient scheme of Construction \ref{con:MainConstruction}. Our simulations were run on three recursively-defined test patterns: the Sierpinski triangle pattern $\S$, the Sierpinski carpet pattern $\C$, and a more complicated pattern which we call {Pattern $\W$}. See Fig. \ref{fig:NumericallySelfSimilarFractals} for an illustration of each of these patterns which are defined formally as follows:

\begin{itemize}\itemsep .5em
\item The {\em Sierpinski triangle pattern}, denoted by $\S$ and illustrated in Fig. \ref{fig:SierpinskiTriangle}, is the $(2,2)$-recursively defined pattern on $\{0,1\}$ where $\S(x,y)=(\S(x,y \text{\,-} 1) + \S(x \text{\,-} 1,y))\!\!\!\mod 2$ when $x,y > 0$; $1$ when $x=0$ or $y=0$; and $\bot$ otherwise.
\item The {\em Sierpinski carpet pattern}, denoted by $\C$ and illustrated in Fig. \ref{fig:SierpinskiCarpet}, is the $(2,2)$-recursively defined pattern on $\{0,1,2\}$ where $\C(x,y) =  (\C(x \text{\,-} 1,y \text{\,-} 1) + \C(x,y \text{\,-} 1) + \C(x \text{\,-} 1,y))\!\!\!\mod 3$ when $x,y >0$; 1 when $x=0$ or $y=0$; and $\bot$ otherwise.
\item {\em Pattern $\W$}, illustrated in Fig. \ref{fig:PatternW}, is the $(3,3)$-recursively defined pattern on $\{0,1\}$ where $\W(x,y)=(\W(x \text{\,-} 2,y \text{\,-} 2) + \W(x,y \text{\,-} 2)\ + \W(x \text{\,-} 1,y \text{\,-} 1) + \W(x \text{\,-} 2,y))\!\!\!\mod 2$ when $x,y > 0$; 1 when $x=0$ and $y$ is even or $x$ is even and $y=0$; 0 when $x=0$ and $y$ is odd or $x$ is odd and $y=0$; and $\bot$ otherwise.
\end{itemize}

For each of these patterns $P$, we ran simulations of the TAS created by Construction \ref{con:KautzLathrop}, denoted by $\T_P$, and the TAS created by Construction \ref{con:MainConstruction}, denoted by $\R_P$. For the Sierpinski triangle, we also ran simulations for the TAS created by the compact error-resilient scheme introduced in \cite{reif2006compact}, denoted by $\B_\S$. Each simulation used a target assembly of $512 \times 512$ tiles, and was allowed to run until the assembly reached 75\% of the target size. 
For each simulation, we computed the number of tiles $N$ in the largest $m \times n$ aggregate assembled (including the origin) without any permanent error. 
Fig. \ref{fig:SimulationResults} shows, for each TAS and each $G_{se}$ value in the range $4.9, 5.2, \ldots, 8.2$, the median value of $N$ for 101 simulations. As suggested in \cite{winfree2003proofreading}, we used a value for $G_{mc}$ slightly less than $2G_{se}$. As can be seen from the figure, our results for the Sierpinski triangle pattern are comparable to the results of \cite{reif2006compact}, and our results for the Sierpinski carpet and $\W$ patterns also show a dramatic improvement over their non error-resilient counterparts. We refer the reader to the website \cite{www} for more information on our simulations.

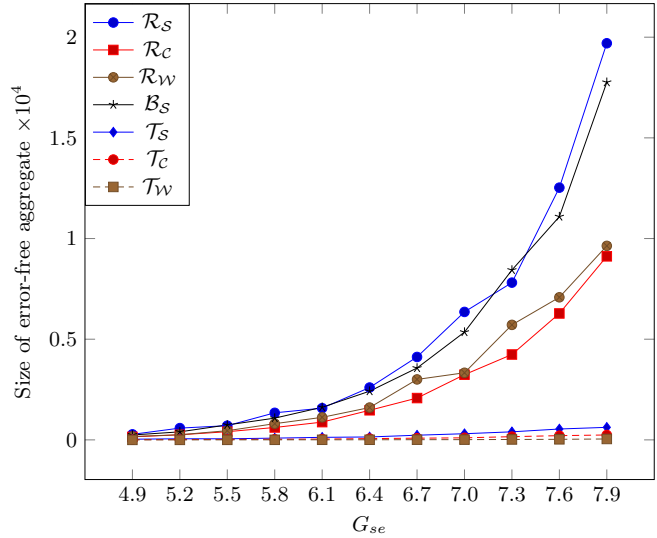
\begin{figure}[!t]
\centering
        \begin{tikzpicture}[scale=.899] 
		\begin{axis}[
			title style={at={(0.5,1.1)}},
			xlabel={$G_{se}$},
			xtick={4.9, 5.2, 5.5, 5.8, 6.1, 6.4, 6.7, 7.0, 7.3, 7.6, 7.9},
			xticklabel style={/pgf/number format/.cd,fixed zerofill,precision=1},
			ylabel={Size of error-free aggregate $\times 10^4$},
			ytick scale label code/.code={},
			legend entries={$\R_\S$,  $\R_\C$, $\R_\W$,  $\B_\S$, $\T_\S$, $\T_\C$, $\T_\W$},
			legend style={anchor=north west,at={(0,1)}}
		]
		
		\addplot table {Simulations/TriangleER.dat};
		\addplot table {Simulations/CarpetER.dat};
		\addplot table {Simulations/W3H3P2M10101010ER.dat};
		\addplot table {Simulations/TriangleReif.dat};
		\addplot table {Simulations/Triangle.dat};
		\addplot table {Simulations/Carpet.dat};	
		\addplot table {Simulations/W3H3P2M10101010.dat};

		\end{axis}
        \end{tikzpicture}
\caption{Simulation results.}
\label{fig:SimulationResults}        
\end{figure}

\section{Discussion}
\label{sec:Discussion}

We leave open the question of whether there exists a redundancy-based compact error-resilient tiling scheme for recursively defined patterns that reduces the mismatch error rate to $\epsilon^3$. Such a scheme exists for Boolean patterns where $op1$ and $op2$ are pairwise input sensitive \cite{sahu2010capabilities}, but it is also conjectured in \cite{sahu2010capabilities} that no such scheme exists when the operations are not pairwise input sensitive.  The conjecture is proven for the case when $op_1$ and $op_2$ are allowed to be arbitrary Boolean operations. 

We remark here that any rectilinear TAS $\T$ can be transformed into a $(2,2)$-recursively defined TAS $\T'$ that can then be transformed into an error resilient TAS using Construction \ref{con:MainConstruction} given in Section \ref{sec:Results} by allowing the label of a tile in $\T'$ to be a pair whose first element represents the label of the original tile in $\T$, and whose second element represents the tile type of the original tile in $\T$. 
The first element is the value used to determine the pattern of $\T'$, and the second element is the value that is communicated in the $V$ and $H$ values of the rectilinear TAS. However, with this approach, the number of tile types used by $\T'$ can be as high as $n^3$ in the general case, where $n$ is the number of distinct tile types used by $\T$, making this approach impractical for TASs with even modestly sized tile sets.

Although our error-resilient TASs performed well in the simulations summarized in Section \ref{sec:Simulations}, further inspection of the data showed that many of the errors in the simulations were caused by boundary tiles attaching in the interior of the assembly. It is likely that better error-resiliency would result from implementing our error-resilient TASs in two stages, one stage for assembling the boundary using only the boundary tile types, and one stage for assembling the interior using only the interior tile types, thereby preventing boundary tiles from attaching in the interior of the assembly. 

%


\vspace{-.5em}

\bibliographystyle{References/splncs04}
\bibliography{References/CESARPBib}

\newpage

\appendix
\section{Technical Appendix}

\subsection{Proof of Lemma \ref{lem:SamePatternAndSize}.}

Before proving the lemma, we first define an operation on tuples. Let $t$ be an $m$-tuple of $n$-tuples, and let $u$ be an $m$-tuple.  Then, $t \wedge u$ is the $m$-tuple of $(n+1)$-tuples defined by
$$
	( (t_{1,1}, \ldots, t_{1,n}, u_1), \ldots, (t_{m,1}, \ldots, t_{m,n}, u_m) )\, .
$$

We now prove Lemma \ref{lem:SamePatternAndSize}. For completeness, it is restated here.

\begingroup
\def\thelemma{\ref{lem:SamePatternAndSize}}
\begin{lemma}
Let $\T=(T,\sigma_\T,2)$ be a TAS created by Construction \ref{con:KautzLathrop}, and let $\R=(R,\sigma_\R,2)$ be the TAS created by Construction \ref{con:MainConstruction} given $\T$. Then $P(\R)=P(\T)$ and $|R|=|T|$.
\end{lemma}
\addtocounter{lemma}{-1}
\endgroup

\begin{proof}[of Lemma \ref{lem:SamePatternAndSize}]

Consider a tile type $t \in T$, and the tile type $r_t \in R$ created by Construction \ref{con:MainConstruction} to replace $t$. By Construction \ref{con:MainConstruction}, $\tlabel(t)=\tlabel(r_t)$, $\gcolor_t(\south) = \first(\gcolor_{r_t}(\south)) \wedge \second(\gcolor_{r_t}(\south))$, and $\gcolor_t(\west)=\second(\gcolor_{r_t}(\west))$. 

	We first show that the function $f: T \to R$ defined by $f(t) = r_t$ for all $t \in T$ is a bijection and hence $|T|=|R|$. Since each tile type $r \in R$ is created from a tile type $t \in T$, it is clear that $f$ is onto. To see that $f$ is 1-1, let $t_1,t_2 \in T$ such that $f(t_1)=f(t_2)=r$. Then, $\tlabel(t_1)=\tlabel(t_2)=\tlabel(r)$, $\gcolor_{t_1}(\south)=\gcolor_{t_2}(\south) = \first(\gcolor_{r}(\south)) \wedge \second(\gcolor_{r}(\south))$, and $\gcolor_{t_1}(\west)=\gcolor_{t_2}(\west)=\second(\gcolor_{r}(\west))$. It follows by Construction \ref{con:KautzLathrop} that $t_1$=$t_2$. Hence, $f$ is 1-1.
	
	To see that $P(\R)=P(\T)$, we can assume that $\T$ was created by Construction \ref{con:KautzLathrop} for a $(w,h)$-recursively defined pattern $P$ for some integers $w,h \ge 2$. Now let $(x,y) \in \Z^2$ such that $x \ge 0$ and $y \ge 0$, and let $t = \T[x,y]$. It suffices to show that $\R[x,y] = r_t$. Since $t = \T[x,y]$, by Construction \ref{con:KautzLathrop}, $\gcolor_t(\west) = \row^{w-1}_P(x-1,y)$ and $\gcolor_t(\south) = \block^{w,h-1}_P(x,y-1)$. 
	Then, by Construction \ref{con:MainConstruction}, $\R[x,y]$ is the tile $r \in R$ such that $\block^{w,h-1}_P(x,y-1) = \first(\gcolor_r(\south)) \wedge \second(\gcolor_r(\south))$ and $\row^{w-1}_P(x-1,y) = \second(\gcolor_r(\west))$. Since the function $f$ is 1-1, there is only one such tile, namely $r_t$.
\qed
\end{proof}

\subsection{Proof of Lemma \ref{lem:NoUVWMismatches}.}

Here we prove Lemma \ref{lem:NoUVWMismatches}. For completeness, it is restated here.

\begingroup
\def\thelemma{\ref{lem:NoUVWMismatches}}
\begin{lemma}
Let $\R$ be a TAS created by Construction \ref{con:MainConstruction} and let $(x,y) \in \Z^2$ where $x>0$ and $y>0$. Let $u=\T[x-1,y-1]$, $v=\T[x,y-1]$ and $w=\T[x-1,y]$. If there is no mismatch error between $u$ and $v$, and no mismatch error between $u$ and $w$, then all of the following equalities hold:
\begin{align*}
	\first(\gcolor_v(\north)) & \overset{(1)}{=} \first(\gcolor_w(\east))\\	
	\first(\gcolor_v(\north)) \dshift \second(\gcolor_v(\north)) & \overset{(2)}{=} \first(\gcolor_v(\east)) \shift \second(\gcolor_v(\east))\\
	\first(\gcolor_w(\east)) \shift \second(\gcolor_w(\east)) & \overset{(3)}{=} \first(\gcolor_w(\north)) \dshift \second(\gcolor_w(\north))\,.
\end{align*}
\end{lemma}
\addtocounter{lemma}{-1}
\endgroup

\begin{figure}[!t]
\centering
	\begin{tikzpicture}[scale=1]
	 
	\draw[step=2cm] (0,0) grid (2,2);
	\draw[step=2cm] (2,0) grid (4,2);
	\draw[step=2cm] (0,2) grid (2,4);
	
	 \node at (1,1) {$\tlabel(u)$} ;
	 \node at (3,1) {$\tlabel(v)$} ;
	 \node at (1,3) {$\tlabel(w)$} ;
	 
	\node[label={[label distance=-2mm]right:\rotatebox{90}{{\normalsize $(u_{b}, u_{r})$}}}] at (0,1) {};
	\node[label={[label distance=-2mm]right:\rotatebox{90}{{\normalsize $(*, w_{r})$}}}] at (0,3) {};
	  
	\node[label={[label distance=-2mm]above:\rotatebox{0}{{\normalsize $(u_{b}, u_{c})$}}}] at (1,0) {};
	\node[label={[label distance=-2mm]above:\rotatebox{0}{{\normalsize $(*, v_{c})$}}}] at (3,0) {};

	\end{tikzpicture}
\caption{The tiles $u$, $v$ and $w$ with no mismatch errors.}
\label{fig:NoUVWMismatches}
\end{figure}
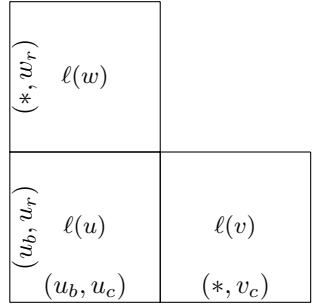

\begin{proof}[of Lemma \ref{lem:NoUVWMismatches}]
We use the following tuples encoded into the glues on the  ``input'' sides of the tiles $u$, $v$ and $w$:
\begin{align*}
	w_{r} &=\second(\gcolor_w(\west))\\
	u_{r} &= \second(\gcolor_u(\west))\\
	u_{b} &= \first(\gcolor_u(\south))=\first(\gcolor_u(\west))\\
	u_{c} &=\second(\gcolor_u(\south))\\
	v_{c} &=\second(\gcolor_v(\south))
\end{align*}
See Fig. \ref{fig:NoUVWMismatches} for an illustration.
We then have that
\begin{align*}
	\first(\gcolor_v(\north)) &\overset{\text{(a)}}{=} \first(\gcolor_v(\west)) \shift \second(\gcolor_v(\west))\\
					 &\overset{\text{(b)}}{=} \first(\gcolor_u(\east)) \shift \second(\gcolor_u(\east))\\
					 &\overset{\text{(c)}}{=} [u_{b}  \dshift u_{c}] \shift [u_{r} \shift \ell(u)]\,,\\
	\second(\gcolor_v(\north)) &\overset{\text{(d)}}{=} v_{c} \shift \tlabel(v)\,,\\
	\first(\gcolor_v(\east))  &\overset{\text{(e)}}{=} \first(\gcolor_v(\west)) \dshift v_{c}\\
					  &\overset{\text{(f)}}{=} \first(\gcolor_u(\east)) \dshift v_{c} \\
					  &\overset{\text{(g)}}{=} [u_{b} \dshift u_{c}] \dshift v_{c}\,,\\
	\second(\gcolor_v(\east)) &\overset{\text{(h)}}{=} \second(\gcolor_v(\west)) \shift \tlabel(v)\\
						&\overset{\text{(i)}}{=} \second(\gcolor_u(\east)) \shift \tlabel(v)\\
					        &\overset{\text{(j)}}{=} [u_{r} \shift \tlabel(u)] \shift \tlabel(v)\,,\\
	\first(\gcolor_w(\north))  &\overset{\text{(k)}}{=} \first(\gcolor_w(\south)) \shift w_{r}\\
				 	     &\overset{\text{(l)}}{=} \first(\gcolor_u(\north)) \shift w_{r}\\
					     &\overset{\text{(m)}}{=} [u_{b} \shift u_{r}] \shift w_{r}\,,\\
	\second(\gcolor_w(\north)) &\overset{\text{(n)}}{=} \second(\gcolor_w(\south)) \shift \tlabel(w)\\
					         &\overset{\text{(o)}}{=} \second(\gcolor_u(\north)) \shift \tlabel(w)\\
						&\overset{\text{(p)}}{=} [u_{c} \shift \tlabel(u)] \shift \tlabel(w)\,,\\
	\first(\gcolor_w(\east)) &\overset{\text{(q)}}{=} \first(\gcolor_w(\south)) \dshift \second(\gcolor_w(\south))\\
					   &\overset{\text{(r)}}{=}  \first(\gcolor_u(\north)) \dshift \second(\gcolor_u(\north))\\
					  &\overset{\text{(s)}}{=} [u_{b} \shift u_{r}] \dshift [u_{c} \shift \tlabel(u)]\,, and\\
	\second(\gcolor_w(\east)) &\overset{\text{(t)}}{=} w_{r} \shift \tlabel(w)
\end{align*}
where equalities
(a), (c), (d), (e), (g), (h), (j), (k), (m), (n), (p), (q), (s) and (t) follow from Construction \ref{con:MainConstruction}; 
equalities (b), (f) and (g) follow from the assumption that there is no mismatch error between $u$ and $v$; and
equalities (i), (o) and (r) follow from the assumption that there is no mismatch error between $u$ and $w$.

We can assume, w.l.o.g., that the TAS $\R$ is created by Construction \ref{con:MainConstruction} from a TAS $\T$ created by Construction \ref{con:KautzLathrop} for a $(m,n)$-recursively defined pattern $P$ for some fixed integers $m,n \ge 2$, and hence, we can further define $u_r$, $u_b$, $u_c$, $v_c$ and $w_r$ as follows:
\begin{align*}
	w_{r} &= (q_1, \ldots, q_{m-1})\\
	u_{r} &= (r_1, \ldots, r_{m-1})\\
	u_{b} &= ( (b_{1,1}, \ldots, b_{1,m-1}), \ldots, (b_{n-1,1}, \ldots, b_{n-1,m-1}) )\\
	u_{c} &= (c_1, \ldots, c_{n-1})\\
	v_{c} &= (d_1, \ldots, d_{n-1})\,.
\end{align*}

{\em Proof of equality (1):} 
We must show that $\first(\gcolor_v(\north)) = \first(\gcolor_w(\east))$. 
By equality (a) above, $\first(\gcolor_v(\north))=[u_{b} \dshift u_{c}] \shift [u_{r} \shift \tlabel(u)]$. 
By equality (s) above, $\first(\gcolor_w(\east))=[u_{b} \shift u_{r}] \dshift [u_{c} \shift \tlabel(u)]$.
Thus, it suffices to show that $[u_{b} \dshift u_{c}] \shift [u_{r} \shift \tlabel(u)] = [u_{b} \shift u_{r}] \dshift [u_{c} \shift \tlabel(u)]$.
This is a straightforward exercise, using the definitions of\ \,$\dshift$\ \,and\ \,$\shift$\ \ given in Section \ref{sec:Preliminaries}:
\begin{align*}
&[u_{b} \dshift u_{c}] \shift [u_{r} \shift \tlabel(u)]\\
= &[u_{b} \dshift u_{c}] \shift [(r_1, \ldots, r_{m-1}) \shift \tlabel(u)]\\
= &[u_{b} \dshift u_{c}] \shift (r_2, \ldots, r_{m-1}, \tlabel(u))\\
= &[( (b_{1,1}, \ldots, b_{1,m-1}), \ldots, (b_{n-1,1}, \ldots, b_{n-1,m-1}) ) \dshift (c_1, \ldots, c_{n-1})]\\
    &\shift (r_2, \ldots, r_{m-1}, \tlabel(u))\\
= &( (b_{1,2}, \ldots, b_{1,m-1}, c_1), \ldots, (b_{n-1,2}, \ldots, b_{n-1,m-1}, c_{n-1}) ) \shift (r_2, \ldots, r_{m-1}, \tlabel(u))\\
= &( (b_{2,2}, \ldots, b_{2,m-1}, c_2), \ldots, (b_{n-1,2}, \ldots, b_{n-1,m-1}, c_{n-1}) , (r_2, \ldots, r_{m-1}, \tlabel(u)) )\,,
\end{align*}
and
\begin{align*}
&[u_{b} \shift u_{r}] \dshift [u_{c} \shift \tlabel(u)]\\
= &[u_{b} \shift u_{r}] \dshift [(c_1, \ldots, c_{n-1}) \shift \tlabel(u)]\\
= &[u_{b} \shift u_{r}] \dshift (c_2, \ldots, c_{n-1}, \tlabel(u))\\
= &[( (b_{1,1}, \ldots, b_{1,m-1}), \ldots, (b_{n-1,1}, \ldots, b_{n-1,m-1}) ) \shift  (r_1, \ldots, r_{m-1})]\\
    &\dshift (c_2, \ldots, c_{n-1}, \tlabel(u))\\
= &( (b_{2,1}, \ldots, b_{2,m-1}), \ldots, (b_{n-1,1}, \ldots, b_{n-1,m-1}), (r_1, \ldots, r_{m-1}))\\
    & \dshift (c_2, \ldots, c_{n-1}, \tlabel(u))\\
= &( (b_{2,2}, \ldots, b_{2,m-1}, c_2), \ldots, (b_{n-1,2}, \ldots, b_{n-1,m-1}, c_{n-1}) , (r_2, \ldots, r_{m-1}, \tlabel(u)) )\,.
\end{align*}

{\em Proof of equality (2):}
We must show that $\first(\gcolor_v(\north)) \dshift \second(\gcolor_v(\north)) = \first(\gcolor_v(\east)) \shift \second(\gcolor_v(\east))$.
By equalities (c) and (d) above, we have that
\begin{align*}
	\first(\gcolor_v(\north)) \dshift \second(\gcolor_v(\north)) = [[u_{b} \dshift u_{c} ] \shift [u_{r} \shift \tlabel(u)]] \dshift [v_{c} \shift \tlabel(v)]
\end{align*}
and, by equalities (g) and (j) above, we have that 
\begin{align*}
	\first(\gcolor_v(\east)) \shift \second(\gcolor_v(\east)) = [[u_{b} \dshift u_{c}] \dshift v_{c}] \shift [[u_{r} \shift \tlabel(u)] \shift \tlabel(v)]\,.
\end{align*}
Thus, it suffices to show that $[[u_{b} \dshift u_{c} ] \shift [u_{r} \shift \tlabel(u)]] \dshift [v_{c} \shift \tlabel(v)] = [[u_{b} \dshift u_{c}] \dshift v_{c}] \shift [[u_{r} \shift \tlabel(u)] \shift \tlabel(v)]$.
This is a straightforward exercise, using the definitions of\ \,$\dshift$\ \,and\ \,$\shift$\ \ given in Section \ref{sec:Preliminaries}:
\begin{align*}
&[[u_{b} \dshift u_{c} ] \shift [u_{r} \shift \tlabel(u)]] \dshift [v_{c} \shift \tlabel(v)] \\
=&[[u_{b} \dshift u_{c} ] \shift [u_{r} \shift \tlabel(u)]]  \dshift [(d_1, \ldots, d_{n-1}) \shift \tlabel(v)] \\
=&[[u_{b} \dshift u_{c} ] \shift [u_{r} \shift \tlabel(u)]]  \dshift (d_2, \ldots, d_{n-1}, \tlabel(v)) \\
=&( (b_{2,2}, \ldots, b_{2,m-1}, c_2), \ldots, (b_{n-1,2}, \ldots, b_{n-1,m-1}, c_{n-1}) , (r_2, \ldots, r_{m-1}, \tlabel(u)) )\\
  & \dshift (d_2, \ldots, d_{n-1}, \tlabel(v)) \\
=&( (b_{2,3}, \ldots, b_{2,m-1}, c_2,d_2), \ldots, (b_{n-1,3}, \ldots, b_{n-1,m-1}, c_{n-1},d_{n-1}),\\
   & (r_3, \ldots, r_{m-1}, \tlabel(u), \tlabel(v) )\,,
\end{align*}
and
\begin{align*}
&[[u_{b} \dshift u_{c}] \dshift v_{c}] \shift [[u_{r} \shift \tlabel(u)] \shift \tlabel(v)]\\
=&[[u_{b} \dshift u_{c}] \dshift v_{c}] \shift [[(r_1, \ldots, r_{m-1}) \shift \tlabel(u)] \shift \tlabel(v)]\\
=&[[u_{b} \dshift u_{c}] \dshift v_{c}] \shift (r_3, \ldots, r_{m-1}, \tlabel(u), \tlabel(v))\\
=&[( (b_{1,2}, \ldots, b_{1,m-1}, c_1), \ldots, (b_{n-1,2}, \ldots, b_{n-1,m-1}, c_{n-1}) ) \dshift v_{c}]\\
   & \shift (r_3, \ldots, r_{m-1}, \tlabel(u), \tlabel(v))\\
=&[( (b_{1,2}, \ldots, b_{1,m-1}, c_1), \ldots, (b_{n-1,2}, \ldots, b_{n-1,m-1}, c_{n-1}) ) \dshift (d_1, \ldots, d_{n-1})]\\
   & \shift (r_3, \ldots, r_{m-1}, \tlabel(u), \tlabel(v))\\
=&( (b_{1,3}, \ldots, b_{1,m-1}, c_1, d_1), \ldots, (b_{n-1,3}, \ldots, b_{n-1,m-1}, c_{n-1}, d_{n-1}) )\\
   & \shift (r_3, \ldots, r_{m-1}, \tlabel(u), \tlabel(v))\\
 =&( (b_{2,3}, \ldots, b_{2,m-1}, c_2, d_2), \ldots, (b_{n-1,3}, \ldots, b_{n-1,m-1}, c_{n-1}, d_{n-1}),\\
    & (r_3, \ldots, r_{m-1}, \tlabel(u), \tlabel(v))\,.
\end{align*}

{\em Proof of equality (3):}
We must show that $\first(\gcolor_w(\east)) \shift \second(\gcolor_w(\east)) = \first(\gcolor_w(\north)) \dshift \second(\gcolor_w(\north))$.
By equlities (s) and (t) above, we have that
\begin{align*}
	\first(\gcolor_w(\east)) \shift \second(\gcolor_w(\east)) =  [[u_{b} \shift u_{r} ] \dshift [u_{c} \shift \tlabel(u)]] \shift [w_{r} \shift \tlabel(w)]
\end{align*}
and, by equalities (m) and (p) above, we have that
\begin{align*}
	\first(\gcolor_w(\north)) \dshift \second(\gcolor_w(\north)) =  [[u_{b} \shift u_{r}] \shift w_{r}] \dshift [[u_{c} \shift \tlabel(u)] \shift \tlabel(w)]\,.
\end{align*}
Thus, it suffices to show that $[[u_{b} \shift u_{r} ] \dshift [u_{c} \shift \tlabel(u)]] \shift [w_{r} \shift \tlabel(w)] =  [[u_{b} \shift u_{r}] \shift w_{r}] \dshift [[u_{c} \shift \tlabel(u)] \shift \tlabel(w)]$.
This is a straightforward exercise, using the definitions of\ \,$\dshift$\ \,and\ \,$\shift$\ \ given in Section \ref{sec:Preliminaries}:
\begin{align*}
&[[u_{b} \shift u_{r} ] \dshift [u_{c} \shift \tlabel(u)]] \shift [w_{r} \shift \tlabel(w)]\\
=&[[u_{b} \shift u_{r} ] \dshift [u_{c} \shift \tlabel(u)]] \shift [(q_1, \ldots, q_{m-1}) \shift \tlabel(w)]\\
=&[[u_{b} \shift u_{r} ] \dshift [u_{c} \shift \tlabel(u)]] \shift (q_2, \ldots, q_{m-1}, \tlabel(w))\\
=& ( (b_{2,2}, \ldots, b_{2,m-1}, c_2), \ldots, (b_{n-1,2}, \ldots, b_{n-1,m-1}, c_{n-1}) , (r_2, \ldots, r_{m-1}, \tlabel(u)) )\\
   & \shift (q_2, \ldots, q_{m-1}, \tlabel(w))\\
=& ( (b_{3,2}, \ldots, b_{3,m-1}, c_3), \ldots, (b_{n-1,2}, \ldots, b_{n-1,m-1}, c_{n-1}) , (r_2, \ldots, r_{m-1}, \tlabel(u)),\\
   & (q_2, \ldots, q_{m-1}, \tlabel(w)))\,,
\end{align*}
and
\begin{align*}
& [[u_{b} \shift u_{r}] \shift w_{r}] \dshift [[u_{c} \shift \tlabel(u)] \shift \tlabel(w)]\\
=& [[u_{b} \shift u_{r}] \shift w_{r}] \dshift [[(c_1, \ldots, c_{n-1}) \shift \tlabel(u)] \shift \tlabel(w)]\\
=& [( (b_{2,1}, \ldots, b_{2,m-1}), \ldots, (b_{n-1,1}, \ldots, b_{n-1,m-1}), (r_1, \ldots, r_{m-1})) \shift (q_1, \ldots, q_{m-1})]\\
   & \dshift (c_3, \ldots, c_{n-1}, \tlabel(u), \tlabel(w))\\
=& [( (b_{3,1}, \ldots, b_{3,m-1}), \ldots, (b_{n-1,1}, \ldots, b_{n-1,m-1}), (r_1, \ldots, r_{m-1}), (q_1, \ldots, q_{m-1}))\\
   & \dshift (c_3, \ldots, c_{n-1}, \tlabel(u), \tlabel(w))\\   
=& ( (b_{3,2}, \ldots, b_{3,m-1}, c_3), \ldots, (b_{n-1,2}, \ldots, b_{n-1,m-1}, c_{n-1}) , (r_2, \ldots, r_{m-1}, \tlabel(u)),\\
   & (q_2, \ldots, q_{m-1}, \tlabel(w)))\,.   
\end{align*}
\qed
\end{proof}

\end{document}